\documentclass[a4paper]{article}
\usepackage[top=2.5cm, bottom=3cm, left=3.6cm , right=3.6cm]{geometry}

\usepackage[cmex10,fleqn]{amsmath}
\usepackage{amsfonts,amssymb,amstext}
\usepackage{hyperref}
\usepackage{authblk}
\usepackage{blkarray}

\usepackage{algorithm,algorithmicx,algpseudocode}
\usepackage{booktabs}
\usepackage{tikz,pgfplots}
\usetikzlibrary{arrows,automata}
\usepackage{todonotes}
\usepackage{subfigure}
\usepackage{multirow}

\title{A Quadratic Assignment Formulation\\of the Graph Edit Distance}
\author[1,4]{S{\'e}bastien Bougleux}
\author[1,4]{Luc Brun}
\author[1,2]{Vincenzo Carletti}
\author[2]{Pasquale Foggia}
\author[3,4]{Benoit~Ga{\"u}z{\`e}re}
\author[2]{Mario Vento}

\affil[1]{GREYC UMR 6072, CNRS - Universit{\'e} de Caen Normandie - ENSICAEN, France}
\affil[2]{MVIA, Dept. of Information, Electrical Engineering and Applied Mathematics, Univ. of Salerno, Italy}
\affil[3]{LITIS, INSA de Rouen, France}
\affil[4]{NormaSTIC FR CNRS 3638, France}

\DeclareMathOperator*{\argmin}{argmin}

\newtheorem{definition}{Definition}

\newtheorem{corollary}{Corollary}
\newtheorem{proposition}{Proposition}

\newtheorem{remarque}{Remark}

\newenvironment{proof}{\noindent{\bf Proof:}}{$\Box$\par}

\DeclareMathOperator{\kgraphs}{k-subgraphs}

\begin{document}
\maketitle
\begin{abstract}
  Computing efficiently a robust measure of similarity or
  dissimilarity between graphs is a major challenge in 
  Pattern Recognition. The Graph Edit Distance (GED) is a flexible measure
  of dissimilarity between graphs which arises in error-tolerant graph matching. It is defined from an optimal sequence of edit operations (edit path) transforming one
  graph into an other. Unfortunately, the exact computation of this
  measure is NP-hard. In the last decade, several approaches have been proposed to approximate the GED in polynomial time, mainly by solving linear programming problems. Among them, the bipartite GED has received much attention. It is deduced from a linear sum assignment of the nodes of the two graphs, which can be efficiently computed by Hungarian-type algorithms. However, edit operations on nodes and edges are not handled simultaneously, which limits the accuracy of the approximation. To overcome this limitation, we propose to extend the linear assignment model to a quadratic one, for directed or undirected graphs having labelized nodes and edges. This is realized through the definition of a family of edit paths induced by assignments between nodes. We formally show that the
  GED, restricted to the paths in this family, is equivalent to a
  quadratic assignment problem. Since this problem is NP-hard, we
  propose to compute an approximate solution by an adaptation of the Integer Projected Fixed Point method. Experiments show that the proposed approach is generally able to reach a more accurate
  approximation of the optimal GED than the bipartite GED, with a computational cost that is still affordable for graphs of non trivial sizes.
\end{abstract}

\section{Introduction}
The definition of efficient similarity or dissimilarity measures
between graphs is a key problem in structural pattern recognition \cite{Conte2004,Foggia2014,Vento2014}. This problem is nicely
addressed by the graph edit distance, which constitutes one of the most flexible
graph dissimilarity measure \cite{tsai74,Bunke1983,Sanfeliu1983,Bunke1999}. Given two graphs $G_1$ and
$G_2$, such a distance may be understood as a measure of the minimal amount of
distortion required to transform $G_1$ into $G_2$.  The graph edit
distance is defined from the notion of edit path which corresponds to
a sequence of elementary transformations of a graph into another. An edit operation is a transformation performed on the structure of a graph, here restricted to be elementary: node or edge insertion, removal and substitution. This is illustrated in Fig.~\ref{fig:edit_path}. Edit operations are
penalized by a real non-negative cost function $c_e(.)$, and the cost of the edit path is defined as the sum of all its elementary operation's
costs. An optimal edit path, transforming a graph $G_1\,{=}\,(V_1,E_1)$ into a graph $G_2\,{=}\,(V_2,E_2)$, have a minimal cost among all edit paths from $G_1$ to $G_2$. Its cost
defines the GED from $G_1$ to $G_2$:
\begin{equation}
  \label{eq:distance-edition}
\text{GED}(G_1,G_2) = \min_{(e_1, \dots, e_k) \in \mathcal{P}(G_1,G_2)} \sum_{i=1}^k c_e(e_i).
\end{equation}
where $\mathcal{P}(G_1,G_2)$ is the set of all edit paths from $G_1$ to $G_2$, and $e_i$ is an edit operation. In this paper, graphs are assumed to be simple and labeled.
\begin{figure}[!t]
\begin{center}
\includegraphics[width=0.6\textwidth]{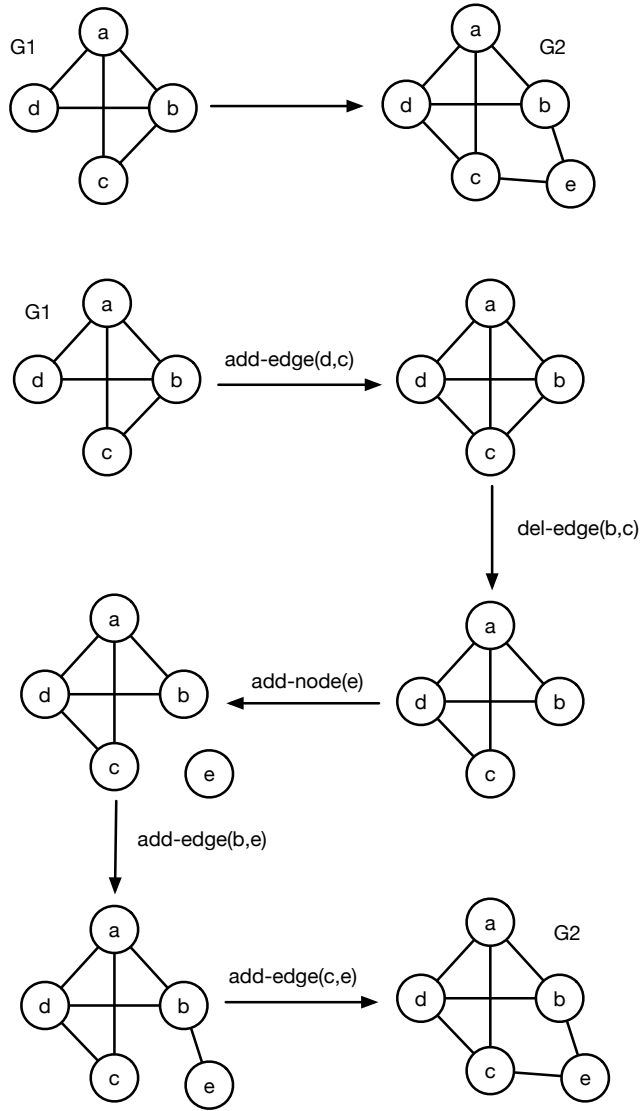}
\end{center}
\caption{A possible edit path to transform the graph $G_1$ into the graph $G_2$. 
If we assume that all the edit operation have an unitary cost, the overall cost of the transformation is equal to 5.}\label{fig:edit_path}
\end{figure}

Computing the GED is NP-hard, in fact NP-complete, and its
approximation is APX-hard \cite{lin94}. A common approach consists in
representing the problem into a state space where the optimal solution
can be found using for example the $A^*$ algorithm, in exponential
time complexity. This is thus restricted to small graphs composed of
about 10 nodes.  Such a complexity compromises the suitability of the
GED in many practical applications where the graphs are usually one
order of magnitude bigger.  However, most of real world problems do
not require the computation of the exact GED, and the use of an
approximation is often sufficient. For this reason, the interest of
the scientific community has been focused on methods providing
efficient approximations of the GED, mainly linear and suboptimal
ones. 

In \cite{Justice2006}, the GED is modeled as 
a binary linear programming problem for graphs with labeled nodes and
unlabeled edges. The relaxation of the initial problem provides a
lower bound of the GED which however cannot be readily associated to
an edit path. An upper bound is also provided through an approximation
of the binary program. The resulting problem corresponds to a square
linear sum assignment problem (LSAP), \textit{i.e.} a weighted
bipartite graph matching, such that the nodes of graph $G_1$ can be
removed, substituted to the nodes of graph $G_2$, as well as the nodes
of $G_2$ can be inserted. A node assignment incorporating possible
removals and insertions is then defined as a bijective mapping, thus
representing a set of edit operations on nodes. Each edit operation is penalized by a cost. The cost of an assignment is then defined as the sum of the costs of its
corresponding edit operations on nodes. The LSAP consists in selecting
an assignment having a minimal cost, which can be computed in
polynomial time complexity, for instance with the Hungarian algorithm,
see \cite{sier01,bur09} for more details on linear programming and
LSAP. The resulting optimal node assignment allows to deduce, in a non ambiguous way, the edge operations that define an edit path, mostly not minimal but short. The cost of such a short path defines an approximate GED.

The same line of research has been followed
in~\cite{Riesen2007,Riesen2009}, but such that labels of edges are also taken into
account in the assignment through the cost of edit operations on nodes, which is neglected in the upper bound proposed in~\cite{Justice2006}. 
The resulting distance, called the
bipartite GED, has received much
attention~\cite{Zeng2009,Fankhauser2011,Serratosa2014,Gauzere2014a,Riesen2014,Riesen2014c,rie15,rie15prai,Ferrer2015,Carletti2015}. In order to select a
relevant assignment, a bag of patterns is attached to each node, and
each possible substitution is penalized by a cost that measures the
affinity between the bags, hence taking into account the edge
information. Similarly, node removals and insertions are
penalized by a cost measuring the importance of the bags. 

The definition of the bags of patterns is a key point, as well as the associated measure of affinity. Incident edges have been initially proposed in \cite{Riesen2007,Riesen2009}, and the cost between two nodes (or bags) is itself defined as the cost of the linear sum assignment of the patterns within the bags, following the same framework as the one defined for the nodes. The cost of substituting, removing and inserting the patterns depends on the original edit cost function $c_e$. The resulting optimal node assignment allows to deduce, in a non ambiguous way, the edge operations that define an edit path, mostly not minimal but short. The cost of such a short path defines the bipartite GED.

Remark that this approach assumes that an edit path may be deduced
from a sequence of edit operations applied on nodes only. As we will
see in this paper, this is possible because there is an equivalence
relation between assignments and edit paths. Intuitively, this is due to
the strong relationship that exists between GED and morphism between
graphs. Under special conditions on the cost of edit operations,
computing the GED is equivalent to compute a maximum common subgraph
of two graphs \cite{Bunke1997,Bunke1999}. More generally, any mapping
between the nodes of two graphs induces an edit path which substitutes
all mapped nodes together with all their incident edges, and inserts
or removes the non-mapped nodes/edges. Conversely, given an edit path
between two graphs, such that each node and each edge is substituted
only once, one can define a mapping between the substituted nodes and
edges of both graphs.

While the bipartite GED provides a good approximation of the GED, it
overestimates it. As shown by several works, this overestimation can
only be marginally reduced, for instance by considering more global
information than the one supported by incident edges
\cite{Zeng2009,Gauzere2014a,Carletti2015}, or by modifying the resulting edit
path by genetic algorithms \cite{Riesen2014}, see \cite{rie15} for
more details.  Although these methods provide an interesting
compromise between time complexity and approximation quality, they are
inherently limited to compute linear approximations of the GED.


To fully describe the GED, both node and edge assignments should be
considered simultaneously. Indeed, operations on edges can only be
deduced from operations performed on their two incident nodes. For
instance, an edge can  be substituted to another one only if its
incident nodes are substituted. This pairwise constraint on nodes is
closely related to the one involved in graph matching. It is known
that graph matching problems, and more generally problems that
incorporate pairwise constraints, can be cast as a quadratic
assignment problem (QAP) \cite{koop57,Lowler1963,law76,Loiola2007,bur09}. QAPs
are NP-hard and so different relaxation algorithms have been proposed
to find an approximate solution, such as Integer Projected Fixed Point
(IPFP) \cite{Leordeanu2009}, or Graduated NonConvexity and Concavity
Procedure \cite{Liu2014}. Even if computing the GED is generally not
equivalent to solving a graph matching problem, it should also be
formalized as a QAP. To the best of our knowledge, this aspect has
only been considered through the definition of fuzzy paths by
\cite{neu07}. Thus, the strong relationships between the GED and the
QAP have not yet been analyzed.

In this paper, we extend the LSAP with insertions and removals \cite{Riesen2007,Riesen2009} to a quadratic one. First, preliminary results concerning edit paths are established (Section~\ref{sec:preliminaries}), allowing to formalize the relation between the LSAP (Section~\ref{sec:linear_formulation}) or the QAP (Section~\ref{sec:qap_formulation}), and such paths. In particular, we show that the GED is a QAP when graphs are simple. Then, we propose an improved IPFP algorithm adapted to the minimization of quadratic functionals to approximate the GED (Section~\ref{sec:qap_solvers}). The approach, validated through experiments in Section~\ref{sec:experiments}, generally provides a more accurate approximation of the exact GED than the bipartite GED, with a computational cost still affordable for graphs of non trivial sizes.


\section{Preliminaries}
\label{sec:preliminaries}
This section introduces some basics about graphs, graph edit distance
and edit paths. We futher introduce different familly of edit paths
and show that one of this familly is in direct correspondence with a
familly of mapping functions between the set of nodes of two graphs.
\subsection{Graph Basics}\label{sec:mainDef}
\begin{definition}[Unlabeled graph]
  An unlabeled  graph $G$ is defined by the couple
  $G=(V,E)$ where $V$ is the set of nodes and $E\subseteq
  V\times V$ is the set of edges. Each edge is an orded couple
  of nodes $(i,j)$ with $i,j \in V$. The direction of an edge is
  implicitly given by the order of its nodes, 
  i.e. the direction of $(i,j)$ is from $i$ to $j$.
\end{definition}
\begin{definition}[Nodes Adjacency]
Given a graph $G=(V,E)$ and two nodes $i,j \in V$. $i$ and $j$ are
  said to be adjacent iff $\exists (i,j) \in E$.
\end{definition}
\begin{definition}[Unlabeled simple graph]
  An unlabeled graph is said to be simple iff:
  \begin{enumerate}
  \item It exists at most one edge between any pair of nodes,
  \item The graph does not contain self loops
    ($(i,i)\not\in E, \forall i \in V$)
  \end{enumerate}
\end{definition}
\begin{definition}[Labeled simple graph]\label{def:labeledsimplegraph}
  Let $\mathcal{L}$ be a finite alphabet of node and edge labels. A
  labeled simple graph  is a tuple $G=(V,E,\mu,\nu)$ where 
  \begin{itemize}
  \item the couple $(V,E)$ defines an unlabeled simple graph,
  \item $\mu : V \rightarrow \mathcal{L}$ is a node labeling function,
  \item $\nu : E\rightarrow \mathcal{L}$ is an edge labeling function.
  \end{itemize}
  The unlabeled graph associated to a given labeled graph
  $G=(V,E,\mu,\nu)$ is defined by the couple $(V,E)$.
\end{definition}

In the following we will only consider simple graphs that we will
simply denote by unlabeled (resp. labeled) graphs.  The term graph
will denote indifferently a labeled or an unlabeled graph.
\begin{definition}[Undirected graph]\label{def:undirectedgraph}
~
\begin{itemize}
  \item A simple graph $G=(V,E)$ is said to be undirected iff\\
  $\forall (i,j) \in E ~ \exists (j,i) \in E$.
  \item A labeled simple graph $G=(V,E,\mu,\nu)$ is said to be undirected iff\\
  $\forall (i,j) \in E ~ \exists (j,i) \in E \land \nu(i,j) = \nu(j,i)$.
\end{itemize}
\end{definition}
\begin{definition}[Bipartite graph]\label{def:bipartitegraphs}
  A graph $G$, labeled or not, is said to be bipartite
  iff \\ $\exists V_1, V_2 \subseteq V : \forall (i,j) \in E, (i \in V_1 \land j \in V_2) 
  \lor (i \in V_2 \land j \in V_1)$.
\end{definition}
\begin{definition}[Subgraph]\label{def:subgraph}
~
  \begin{itemize}
  \item An unlabeled graph $G_1=(V_1,E_1)$ is said to be an unlabeled
    subgraph of $G_2=(V_2,E_2)$ if $V_1\subseteq V_2$ and
    $E_1\subseteq E_2\cap (V_1\times V_1)$. The unlabeled subgraph $G_1$ is 
    called an unlabeled proper subgraph if $V_1\neq V_2$ or $E_1\neq
    E_2$.
  \item If $G_1=(V_1,E_1,\mu_1,\nu_1)$ and $G_2=(V_2,E_2,\mu_2,\nu_2)$
    are both labeled graphs then $G_1$ is a (proper) subgraph of $G_2$
    if
    $(V_1,E_1)$ is an unlabeled (proper) subgraph of $(V_2,E_2)$
    and if the
    following additional constraint is fulfilled: ${\mu_2}_{|V_1}=\mu_1$ and ${\nu_2}_{|E_1}=\nu_1$, where $f_{|}$ denotes the restriction of function $f$ to a
    particular domain.
  \item A structural subgraph of a labeled graph $G$ is an unlabeled
    subgraph of the unlabeled graph associated to $G$.
  \end{itemize}
\end{definition}
\subsection{Edit operations, paths, and distance}\label{sec:mainDefEditDist}
\begin{definition}[Elementary edit operations]\label{def:elEditOp}  
  An elementary edit operation is one of the following operation
  applied on a graph:
  \begin{enumerate}
  \item[$\bullet$] Node/Edge removal. Such removals are defined as the
    removal of the considered element from sets $V$ or $E$.
  \item[$\bullet$] Node/Edge insertion. On labeled graphs, a vertex/edge
    insertion also associates a label to the inserted element.
  \item[$\bullet$] Node/Edge substitution if the graph is a labeled one.
    Such an operation modifies the label of a node or an edge and
    thus transforms the node or edge labeling functions.
  \end{enumerate}
\end{definition}
\begin{definition}[Cost of an elementary edit operation]\label{def:elEdOpCost}
  Each elementary operation $x$ is associated to a cost encoded by a
  specific function for each type of operation:
  \begin{itemize}
  \item Node ($c_{vd}(x)$) and edge removal ($c_{ed}(x)$)
  \item Node ($c_{vi}(x))$ and edge ($c_{ei}(x)$) insertion,
  \item Node ($c_{vs}(x)$) and edge ($c_{es}(x)$) substitution on labeled graphs.
  \end{itemize}
  By extension, we will consider that functions $c_{vd}$ and
  $c_{vi}$ (resp. $c_{ed}$ and $c_{ei}$) apply on the set of nodes
  (resp. the set of edges) of a graph. Hence, the cost $c_{vd}(v)$
  denotes the cost of removing node $v$.

  We assume that a substitution transforming one label into the same
  label has zero cost: 
  \[
  \forall l\in \mathcal{L}, \;c_{vs}(l\rightarrow l)=c_{es}(l\rightarrow l)=0
  \]
  where $l\rightarrow l'$ denotes the substitution of label $l$ into
  $l'$ on some edge or node.
\end{definition}
\begin{definition}[Edit path]\label{def:editPath}
  An edit path of a graph $G$ is a sequence of elementary
    operations applied on $G$, where node removal and edge
    insertion have to satisfy the following constraints:
    \begin{enumerate}
    \item\label{item:vertexremov} A node removal implies a first removal of all its incident
      edges,
    \item\label{item:edgeinsertion} An edge insertion can be applied
      only between two existing or already inserted nodes. 
    \item\label{item:simplegraph} Edge insertions should not create
      more than one edge between two vertices nor create self-loops.
    \end{enumerate}
  An edit path that transforms a graph $G_1$ into a graph $G_2$ is an edit
    path of $G_1$ whose last graph is $G_2$. If $G_1$ and $G_2$ are unlabeled we assume that no node nor edge
  substitutions are performed.
\end{definition}
\begin{definition}[Cost of an edit path]
  The cost of an edit path $P$, denoted $\gamma(P)$ is the sum of
  the costs of its elementary edit operations.
\end{definition}
\begin{definition}[Edit distance]
  The edit distance from a graph $G_1$ to a graph $G_2$ is defined as
  the minimal cost of all edit paths from $G_1$ to $G_2$.
  \[
  d(G_1,G_2)=\min_{P\in \mathcal{P}(G_1,G_2)} \gamma(P)
  \]
  where $\mathcal{P}(G_1,G_2)$ is the set of all edit paths
  transforming $G_1$ into $G_2$. An edit path from $G_1$ to $G_2$ with a
  minimal cost is called an optimal path.
\end{definition}
\begin{proposition}
  \label{prop:resultEditPath}
  Given any graph $G$, and any edit path $P$ of $G$, the
  transformation of $G$ by $P$ is still a simple graph. 
\end{proposition}
\begin{proof}
  Let $G=(V,E,\mu,\nu)$ and $G'=(V',E',\mu',\nu')$ denote the
  initial graph and its transformation by $P$. Since the insertion
  of nodes and edges induces the definition of their labels,
  function $\mu'$ (resp. $\nu'$) defines a valid labeling function
  over $V'$ (resp. $E'$). 

  Let us consider $(u,v)\in E'$. Vertices $u$ and $v$ should either be
  present in $V$ or have been inserted before the insertion of edge
  $(u,v)$ (Definition~\ref{def:editPath},
  condition~\ref{item:edgeinsertion}). Moreover, none of these nodes
  can be removed after the last insertion of edge $(u,v)$ since such a
  removal would imply the removal of $(u,v)$
  (Definition~\ref{def:editPath}, condition~\ref{item:vertexremov}).
  Both $u$ and $v$ thus belong to $V'$.
  Hence  $E'\subseteq V'\times V'$. Moreover, according to
  Definition~\ref{def:editPath}, condition~\ref{item:simplegraph} the
  edge $(u,v)$ can not be inserted if it already exists in $G$ and
  $u\neq v$.  It follows that $G'=(V',E',\mu',\nu')$ is a labeled
  simple graph according to Definition~\ref{def:labeledsimplegraph}.
\end{proof}
\begin{definition}[Independent edit path]
  \label{def:mineditpath}
  An independent edit path between two labeled graphs $G_1$ and $G_2$ is
  an edit path such that:
  \begin{enumerate}
  \item\label{item:subremov} No node nor edge is both substituted and
    removed,
  \item\label{item:subins} No node nor edge is simultaneously substituted and inserted,
  \item\label{item:insremov} Any inserted element is never removed,
  \item\label{intem:subsonce} Any node or edge is substituted at most
    once,
  \end{enumerate}    
\end{definition}
Note that an independent edit path is not minimal in the number of
operations. Indeed, Definition~\ref{def:mineditpath} still allows to
replace one substitution by one removal followed by one insertion
(but such an operation can be performed only once for each node or
edge thanks to condition~\ref{item:insremov}).  We however forbid useless
operations such as the substitution of one node followed by its
removal (condition~\ref{item:subremov}) or the insertion of a node
with a wrong label followed by its substitution
(condition~\ref{item:subins}).

In the following we will only
consider independent edit paths that we simply call edit paths.
\begin{proposition}
  The elementary operations of an independent edit path between two graphs
  $G_1$ and $G_2$ may be ordered into a sequence of removals,
  followed by a sequence of substitutions and terminated by a
  sequence of insertions.
\end{proposition}
\begin{proof}
  Let $R, S$ and $I$ denote the sub-sequences of
  Removals, Substitutions and Insertions of an edit path $P$,
  respectively. Since no removal may be performed on a substituted
  element (condition~\ref{item:subremov} of
  Definition~\ref{def:mineditpath}) and no removal may be performed on
  an inserted element (condition~\ref{item:insremov}), removals only
  apply on elements which are neither substituted nor inserted. Such
  removals operations may thus be grouped at the beginning of the edit
  path. Now, since an element cannot be substituted after its
  insertion, substitutions can only apply on the remaining elements after
  the removal step and can be grouped after the removal operations.  The
  remaining operations only contain insertions.

  Let us consider the sequence of elementary operations $(R,S,I)$
  the order within sequences $R$, $S$ and $I$ being deduced from the
  one of $P$. Such a sequence may be defined since operations in $R$
  apply on elements not in $S$ and $I$ while operations in $S$ do
  not apply on the same elements than operations in $I$. Such sets
  are independent, leading to the definition of independent edit
  paths. However, we still have to show that the sequence $(R,S,I)$
  defines a valid edit path.
  \begin{enumerate}
  \item Since $R$ contains all the removal operations contained in
    $P$, if $P$ satisfies condition~\ref{item:vertexremov} of
    Definition~\ref{def:editPath}, so does the sequence $R$.
  \item Let us suppose that an edge insertion is valid in sequence
    $P$ while it violates Definition~\ref{def:editPath},
    condition~\ref{item:edgeinsertion} in sequence $(R,S,I)$. Let us
    denote by $(u,v)$ such an edge. Edge $(u,v)$ violates
    condition~\ref{item:edgeinsertion} in sequence $(R,S,I)$ only if
    either the removal of $u$ or $v$ belongs to $R$.  In such a case
    the insertion of $(u,v)$ in $P$ should be made before the
    removal of $u$ or $v$. But such a removal would imply the
    removal of all the incident edges of $u$ or $v$
    (Definition~\ref{def:editPath},
    condition~\ref{item:vertexremov}) including the newly inserted
    edge $(u,v)$.  Such an operation would violate the independence
    of $P$ (Definition~\ref{def:mineditpath},
    condition~\ref{item:insremov}).
  \end{enumerate}
  The sequence $(R,S,I)$ is thus a valid edit path which transforms
  a graph $G_1$ into $G_2$ if $P$ do so.
  Furthermore, it is readily seen that all the conditions of Definition~\ref{def:mineditpath}
  are satisfied by the sequence $(R,S,I)$ as soon as they are satisfied by $P$.
  The sequence $(R,S,I)$ is thus an independent edit path.
\end{proof}
\begin{proposition}\label{prop:editPathSubGraphs}
  Let $P$ be an edit path between two graphs $G_1$ and $G_2$. Let us
  further denote by $R$, $S$ and $I$ the sequence of node and edge
  Removals, Substitutions and Insertions performed by $P$, the order
  in each sequence being deduced from the one of $P$. Then:
  \begin{itemize}
  \item the graph $\hat{G}_1$ obtained from $G_1$ by applying
    removal operations $R$ is a subgraph of $G_1$,
  \item the graph $\hat{G}_2$ obtained from $G_1$ by applying
    the sequence of operations $(R,S)$ is a subgraph of $G_2$,
  \item Both $\hat{G}_1$ and $\hat{G}_2$ correspond to a same common
    structural subgraph of $G_1$ and $G_2$.
  \end{itemize}
\end{proposition}
\noindent{\bf Proof:}
  \begin{enumerate}
  \item Since the sequence $R$ is an edit path, $\hat{G}_1$ is a
    graph by Proposition~\ref{prop:resultEditPath}. Moreover, since
    $R$ is only composed of removal operations, we trivially have
    $\hat{V}_1\subset V_1$ and $\hat{E}_1 \subset E_1$. The fact that
    $\hat{E}_1\subset E_1\cap \hat{V}_1\times \hat{V}_1 $ is induced
    by the fact that $\hat{G}_1$ is a graph. Moreover, if $G_1$ is a
    labeled graph, since removal operations do not modify labels,
    labels on $\hat{G}_1$ are only the restriction of the ones on
    $G_1$ to $\hat{V}_1$ and $\hat{E}_1$.
  \item The graph $\hat{G}_2$ is deduced from $G_1$ by the edit path
    $(R,S)$, it is thus a graph. Moreover, $G_2$ is deduced from
    $\hat{G}_2$ by the sequence of insertions $I$. We thus trivially
    have: $\hat{V}_2\subset V_2$ and $\hat{E}_2 \subset E_2\cap
    \hat{V}_2\times\hat{V}_2 $.  Moreover, since insertion operations
    do not modify the label of existing elements, the restriction of
    the label functions of $G_2$ to $\hat{V}_2$ and $\hat{E}_2$
    corresponds to the label functions of $\hat{G}_2$.
  \item Sub graph $\hat{G}_2$ is deduced from $\hat{G}_1$ by the
    sequence of substitution operations $S$. Since substitution
    operations only modify label functions, the structure of both
    graphs is the same and there exists a structural isomorphism
    between both graphs.~~{$\Box$}
  \end{enumerate}
One should note that it may exist several structural isomorphisms
between $\hat{G}_1$ and $\hat{G}_2$. The set of substitutions $S$
fixes one of them, say $f$ such that the image of any element of
$\hat{G}_1$ by $f$ have the same label than the one defined by the
substitution. More precisely, let us suppose that we enlarge the set
of substitution $S$ by $0$ cost substitutions so that all the
nodes and edges of $\hat{G}_1=(\hat{V}_1,\hat{E}_1,\mu_1,\nu_1)$
are substituted. In this case, we have:
\[
\left\{  \begin{array}{llll}
    \forall v\in \hat{V}_1, &\mu_2(f(v))&=&l_v\\
    \forall e\in \hat{E}_1, &\nu_2(f(e))&=&l_e\\
  \end{array}\right.
\]
where $l_v$ and $l_e$ correspond to the labels defined by the
substitutions of $v$ and $e$ and $\mu_2$ and $\nu_2$ define
respectively the node and edge labeling functions of $G_2$.
\begin{corollary}\label{coro:costEditPath}
  Using the same notations than in
  Proposition~\ref{prop:editPathSubGraphs}, the cost $\gamma(P)$ of an
  edit path is defined by:
  \renewcommand{\arraystretch}{2}
  \begin{equation*}
    \begin{array}{lcl}
      \gamma(P) & = & \displaystyle   \sum_{v\in V_1\setminus\hat{V}_1} c_{vd}(v)+
      \sum_{e\in E_1\setminus\hat{E}_1} c_{ed}(e)+
      \sum_{v\in \hat{V}_1} c_{vs}(v)+
      \sum_{e\in \hat{E}_1} c_{es}(e) \\
      & & \displaystyle + \sum_{v\in V_2\setminus\hat{V}_2} c_{vi}(v) + \sum_{e\in E_2\setminus\hat{E}_2} c_{ei}(e)
    \end{array}
  \end{equation*} 
\end{corollary}
\noindent{\bf Proof:}
  The edit path $P$ and its rewriting in $(R,S,I)$ have the same set of
  operations and thus a same cost.
  \begin{description}
  \item[From $G_1$ to $\hat{G}_1$:] Operations in $R$ remove
    nodes in $V_1\setminus\hat{V}_1$ and edges in $E_1\setminus\hat{E}_1$.
  \item[From $\hat{G}_1$ to $\hat{G}_2$:] Substitutions of $S$ apply
    between the two graphs $\hat{G}_1$ and $\hat{G}_2$. Let us
    consider the set of substitutions $S'$ which corresponds to the
    completion of $S$ by $0$ cost substitutions so that all nodes
    and edges of $\hat{G}_1$ are substituted. Both $S$ and $S'$ have
    a same cost. The cost of $S'$ is defined as the sum of costs of
    the substituted nodes and edges of $\hat{G}_1$.
  \item[From $\hat{G}_2$ to $G_2$:] Operations in $I$ insert
    nodes of $V_2\setminus\hat{V}_2$ and edges of $E_2\setminus\hat{E}_2$ in
    order to obtain $G_2$ from $\hat{G}_2$.~~{$\Box$}
  \end{description}
\begin{remarque}\label{rem:doubleCounting}
  Using the same notations than
  Proposition~\ref{prop:editPathSubGraphs} if both $G_1$ and $G_2$ are
  undirected we have:
  \[
  \gamma(P)=\gamma_v(P)+\gamma_e(P)
  \]
  with
  \begin{equation*}
    \begin{array}{lcl}
      \gamma_v(P) & = & \displaystyle   \sum_{i\in V_1\setminus\hat{V}_1} c_{vd}(i)+
      \sum_{i\in \hat{V}_1} c_{vs}(i)+\sum_{k\in V_2\setminus\hat{V}_2} c_{vi}(k)\\
      \gamma_e(P)&=&\displaystyle\frac{1}{2}\left(\sum_{(i,j)\in \hat{E}_1} c_{es}((i,j)) +
                     \sum_{(i,j)\in E_1\setminus\hat{E}_1} c_{ed}((i,j))+\sum_{(k,l)\in E_2\setminus\hat{E}_2} c_{ei}((k,l))\right)\\
    \end{array}
  \end{equation*} 
\end{remarque}

  Indeed, if both graphs $G_1$ and $G_2$ are undirected both $(i,j)$
  and $(j,i)$ belong to $E_1$ while encoding a single edge $e$.  The
  removal or the substitution of the edge $e$ is thus counted twice in
  $\gamma_e(P)$. In the same way $(k,l)$ and $(l,k)$ represent the
  same edge $e$ of $E_2\setminus\hat{E_2}$ which is thus inserted twice in
  $\gamma_e(P)$. The factor $\frac{1}{2}$ of $\gamma_e(P)$ removes
   this double couting of edge operations.
\begin{corollary}\label{coro:costEditPathConst}
  If all costs do not depend on the node/edge involved in the operations, \textit{i.e.} edit cost functions $c_{vd}$, $c_{ed}$, $c_{vs}$, $c_{es}$, $c_{vi}$, and $c_{ei}$ are constant, the cost of an edit path $P$ is equal to:
  \renewcommand{\arraystretch}{1.5}
  \begin{equation*}
    \begin{array}{lcl}
      \gamma(P) & = & (|V_1|-|\hat{V}_1|)c_{vd}+(|E_1|-|\hat{E}_1|)c_{ed}+V_fc_{vs}+E_fc_{es}\\
      &   & +\,(|V_2|-|\hat{V}_2|)c_{vi}+(|E_2|-|\hat{E}_2|)c_{ei}    
    \end{array}
  \end{equation*}
  where $V_f$ (resp. $E_f$) denotes the number of
  nodes (resp. edges) substituted with a non-zero cost.
  
  Moreover, minimizing the cost of such an edit path is
  equivalent to maximizing the following formula:
  \[
  M(P)\overset{\text{def.}}{=}|\hat{V}_1|(c_{vd}+c_{vi})+
  |\hat{E}_1|(c_{ed}+c_{ei})-V_fc_{vs}-E_fc_{es}
  \]
\end{corollary}
\begin{proof}
  We deduce immediately from Corollary~\ref{coro:costEditPath} the
  following formula:
  \renewcommand{\arraystretch}{1.5}
  \begin{equation*}
    \begin{array}{lcl}
      \gamma(P) & = & (|V_1|-|\hat{V}_1|)c_{vd}+(|E_1|-|\hat{E}_1|)c_{ed}+V_fc_{vs}+E_fc_{es}\\
      &   & +\,(|V_2|-|\hat{V}_2|)c_{vi}+(|E_2|-|\hat{E}_2|)c_{ei}
    \end{array}
  \end{equation*}
  We obtain by grouping constant terms:
  \renewcommand{\arraystretch}{1.5}
  \begin{equation*}
    \begin{array}{lcl}
      \gamma(P) & = & |V_1|c_{vd}+|E_1|c_{ed}+|V_2|c_{vi}+|E_2|c_{ei}\\
      & & - \left[
        |\hat{V}_1|c_{vd}+|\hat{E}_1|c_{ed}+|\hat{V}_2|c_{vi}+|\hat{E}_2|c_{ei}-V_fc_{vs}-E_fc_{es}\right]
    \end{array}
  \end{equation*}
  Since there is a structural isomorphism between $\hat{G}_1$ and
  $\hat{G}_2$, we have $\hat{V}_1=\hat{V}_2$ and
  $\hat{E}_1=\hat{E}_2$. So:
  \begin{equation*}
    \begin{array}{lcl}
      \gamma(P) & = & |V_1|c_{vd}+|E_1|c_{ed}+|V_2|c_{vi}+|E_2|c_{ei} \\
      & & - \left[
        |\hat{V}_1|(c_{vd}+c_{vi})+
        |\hat{E}_1|(c_{ed}+c_{ei})-V_fc_{vs}-E_fc_{es}\right]
    \end{array}
  \end{equation*}
  The first part of the above equation being constant, the
  minimization of $\gamma(P)$ is equivalent to the maximization of
  the last part of the equation.
\end{proof}
\begin{definition}[Restricted edit path]
  A restricted edit path is an independent edit path in which any node or
  any edge cannot be removed and then inserted.
\end{definition}
The term restricted should be understood as minimal number
of operations. The cost of a restricted edit path may not be minimal
(among all edit paths) if the cost of a removal operation followed by
an insertion is lower than the cost of the associated
substitution. However, such a drawback may be easily solved by setting
a new substitution cost equal to the minimum between the old substitution
cost and the sum of the costs of a removal and an insertion. In this
case all non-zero cost substitutions, for nodes and edges, may be equivalently
replaced by a removal followed by an insertion.
\begin{proposition}\label{prop:edit_path_mapping}
  If $G_1$ and $G_2$ are simple graphs, there is a one-to-one mapping
  between the set of restricted edit paths between $G_1$ and $G_2$ and
  the set of injective functions from a subset of $V_1$ to $V_2$. We
  denote by $\varphi_0$, the special function from the empty set onto
  the empty set.
\end{proposition}
\begin{proof}
  Let $P$ denote an edit path. If no node substitution occurs, all
  node of $G_1$ are first removed and then all nodes of $G_2$
  are inserted. We associate this edit path to $\varphi_0$. 

  If substitutions occur. We associate to $P$ the function previously
  mentioned which maps each substituted node of $V_1$ onto the
  corresponding node of $G_2$. 

  Let $\psi$ denotes this mapping. Let us consider an injective
  function $f\neq \varphi_0$ from a set $\hat{V_1}\subset V_1$ onto $V_2$.  We
  associate to $f$ the sets of node
  \begin{itemize}
  \item removals: $c_{vd}(v\rightarrow \epsilon), v\in V_1-\hat{V_1}$,
  \item insertions:  $c_{vi}(\epsilon\rightarrow v), v\in V_2-f(\hat{V_1})$, and
  \item substitutions: $c_{vs}(v\rightarrow f(v)), v\in \hat{V_1}$.   
  \end{itemize}
  Moreover, since $G_1$ and $G_2$ are simple graphs it exists at most
  one edge between any pair of nodes. We thus define the following
  set of edge operations:
  \begin{itemize}
  \item removals: $c_{ed}((i,j)\rightarrow \epsilon)$ such that $i$ or $j$ does not belong to $\hat{V_1}$ or $(f(i),f(j))$ do not belongs to $E_2$,
  \item insertions $c_{ei}(\epsilon \rightarrow (k,l) )$ such that $k$ or $l$ does not belong to $f(\hat{V_1})$ or $(f^{-1}(k),f^{-1}(l))$ do not belongs to $E_1$,

  \item substitutions $c_{es}((i,j)\rightarrow (f(i),f(j))), (i,j)\in (\hat{V_1}\times\hat{V_1})\cap E_1$ and $(f(i),f(j)))\in E_2$.
  \end{itemize}
  Let us denote respectively by $R,S$ and $I$ the set of
  removals/substitutions/insertions defined both on node and
  edges. The sequence $(R,S,I)$ defines a valid restricted edit path whose image
  by $\psi$ is by construction equal to $f$. The function $\psi$ is
  thus surjective.

  Let us consider two different edit paths $P=(R,S,I)$ and $P'=(R',S',I')$. Then:
  \begin{itemize}
  \item If $R\neq R'$. If node removals are not equal, $\psi(P)$ and
    $\psi(P')$ are not defined on the same set and are consequently
    not equal. Otherwize, let us suppose that an edge $(i,j)$ is
    removed in $P$ and not in $P'$. Let us further suppose that
    $\psi(P)(i)=\psi(P')(i)$ and $\psi(P)(j)=\psi(P')(j)$. Since
    $(i,j)$ is not removed in $P'$ we have
    $(\psi(P)(i),\psi(P)(j))=(\psi(P')(i),\psi(P')(j))\in E_2$. The
    edge $(i,j)$ should thus be first removed and then inserted in $P$
    which contradicts the definition of a restricted edit
    path. Therefore, one of the two following conditions should holds:
    \begin{enumerate}
    \item     $\psi(P)(i)\neq\psi(P')(i)$ or $\psi(P)(j)\neq\psi(P')(j)$,
    \item the set of edge removal operations is the same in $P$ and $P'$.
    \end{enumerate}
    If the first condition holds $\psi(P)\,{\neq}\,\psi(P')$. If the second
    condition holds, since $R\,{\neq}\,R'$, the set of node removals is
    different in $P$ and $P'$. In this case we get also
    $\psi(P)\,{\neq}\,\psi(P')$.
  \item If $S\neq S'$. Node substitutions are different if they are
    either defined on different sets or do not correspond to the same
    node mapping. In both cases, $\psi(P)\neq \psi(P')$. If node
    substitutions are identical but if $S$ and $S'$ differ by the
    edge substitutions, an edge of $G_1$ is substituted by $P$ and $P'$ to an edge of
    $G_2$ with two different labels. Since an edge may
    be substituted at most once in an independent edit path, either
    $P$ or $P'$ is not a valid edit path between $G_1$ and $G_2$.
    Hence $S$ and $S'$ differ only if their node substitutions differ, 
    in which case we have $\psi(P)\neq \psi(P')$.
  \item If $I\neq I'$ (with $R=R'$ and $S=S'$). If a node $k\in V_2$
    is inserted by $I$ but not by $I'$, this means that $k$ is
    substituted by $S'$. Hence we contradict $S=S'$. In the same way,
    if an edge $(k,l)\in E_2$ is inserted by $I$ and not by $I'$,
    $(k,l)$ is substituted by $S'$ but not by $S$. We again contradict
    $S=S'$.   
  \end{itemize}
  Note that the last item of the previous decomposition shows (as a
  side demonstration) that given the initial and final graphs, a
  restricted edit path is fully defined by its set of removals and
  substitutions.  Moreover, if the set of removals or substitutions of
  two restricted edit paths differ, then the associated mapping is
  different. The function $\psi$ is thus injective and hence
  bijective.
\end{proof}
\begin{proposition}\label{prop:defSRIsets}
  Let $P$ be a restricted edit path not associated with $\varphi_0$
  (hence with some substitutions). Let us denote by
  $\varphi: \hat{V_1}\rightarrow V_2$ the injective function
  associated to $P$ and let us denote $\varphi(\hat{V_1})$ by
  $\hat{V_2}$. We further introduce the following two sets:
  \[
  \left\{
    \begin{array}{lll}
      R_{12}&=&\{(i,j)\in E_1\cap(\hat{V_1}\times\hat{V_1})~|~(\varphi(i),\varphi(j))\not\in E_2\}\\
      I_{21}&=&\{(k,l)\in E_2\cap(\hat{V_2}\times\hat{V_2})~|~(\varphi^{-1}(k),\varphi^{-1}(k))\not\in E_1\}\\
    \end{array}
  \right.
  \]
  \begin{itemize}
  \item The set of substituted/removed/inserted nodes by $P$ are
    respectively equal to: $\hat{V_1}$, $V_1\setminus\hat{V_1}$ and
    $V_2\setminus\hat{V_2}$.
  \item The set of edges substituted/removed/inserted by $P$ are
    respectively equal to:
    \begin{itemize}
    \item Substituted: $\hat{E_1}=\left(E_1\cap (\hat{V_1}\times\hat{V_1})\right)\setminus R_{12}$\\
      with $\hat{E_2}=\varphi(\hat{E_1})=\left(E_2\cap (\hat{V_2}\times\hat{V_2})\right)\setminus I_{21}$
    \item Removed: $E_1\setminus\hat{E_1}=\left(E_1\cap\left(((V_1\setminus\hat{V_1})\times V_1)\cup ( V_1\times(V_1\setminus\hat{V_1}))\right)\right)\cup R_{12}$      
    \item Inserted: $E_2\setminus\hat{E_2}=\left(E_2\cap\left(((V_2\setminus\hat{V_2})\times V_2)\cup ( V_2\times(V_2\setminus\hat{V_2}))\right)\right)\cup I_{21}$
    \end{itemize}
  \end{itemize}
\end{proposition}
\noindent{\textbf{Proof:}~}
  Node substitions/removal/insertions are direct consequences of the
  proof of Proposition~\ref{prop:edit_path_mapping} and the bijective
  mapping between injective functions from a subset
  $\hat{V_1}\subset V_1$ onto $V_2$ and restricted edit paths. 
  \begin{itemize}
  \item If
    $(i,j)\in \hat{E_1}, \{\varphi(i),\varphi(j)\}\subset \hat{V_2}$
    and $(\varphi(i),\varphi(j))\in E_2$. Since $(i,j)$ cannot be
    removed and then inserted it must be substituted. Conversely, if
    $(i,j)\in E_1$ is substituted, $i$ an $j$ must be substituted
    ($\{i,j\}\subset \hat{V_1}$). Moreover, the edge
    $(\varphi(i),\varphi(j))$ should exists in $E_2$. Hence
    $(i,j)\not\in R_{12}$ and $(i,j)\in \hat{E_1}$. 

    If $(k,l)\in \hat{E_2}=\varphi(\hat{E_1})$, it exits
    $(i,j)\in E_1$ such that $k=\varphi(i)$ and $l=\varphi(j)$. Hence
    $(k,l)\not\in I_{21}$. Since $(i,j)\not\in R_{12}, (k,l)\in E_2$
    and $\{k,l\}\subset \hat{V_2}$. Hence
    $(k,l)\in \varphi(\hat{E_1})=E_2\cap
    (\hat{V_2}\times\hat{V_2})\setminus I_{21}$.
    Conversely, if
    $(k,l)\in E_2\cap (\hat{V_2}\times\hat{V_2})\setminus I_{21}$, it exits
    $(i,j)$ such that $k=\varphi(i)$ and $l=\varphi(j)$
    ($\{k,l\}\in \hat{V_2}$). Since $(k,l)\in E_2\setminus I_{21}$ we have
    $(i,j)\in E_1\setminus R_{12}$. Finally, $\{k,l\}\subset \hat{V_2}$ implies
    that $\{i,j\}\subset \hat{V_1}$ hence $(i,j)\in \hat{E_1}$ and
    $(k,l)\in \hat{E_2}=\varphi(\hat{E_1})$.
  \item Any non substited edge of $G_1$ must be removed. Hence, the
    set of removed edges is equal to $E_1\setminus\hat{E_1}$. The remaining
    equation, is deduced by a negation of the conditions defining
    $\hat{E_1}$.
  \item In the same way, any edge of $G_2$ which is not produced by a
    substitution of an edge of $G_1$ must be inserted. Hence, the set
    of inserted edges is equal to $E_2\setminus\hat{E_2}$. The negation of the
    condition defining $\hat{E_2}$ provides the remaining equation.~~$\Box$
  \end{itemize}
\subsection{Linear and quadratic assignment problems}\label{sec:mainDefAssignment}
Within this report, an assignment corresponds to a bijective mapping $\varphi:\mathcal{X}{\rightarrow}\mathcal{Y}$ between two finite sets
$\mathcal{X}$ and $\mathcal{Y}$ having the same size $|\mathcal{X}| = |\mathcal{Y}| = n$.
When these sets are reduced to the same set of integers, \textit{i.e.} $\mathcal{X}\,{=}\,\mathcal{Y}\,{=}\,\{1,\ldots,n\}$, the bijection $\varphi$ reduces to the permutation $(\varphi(1),\ldots,\varphi(n))$. Any permutation $\varphi$ can be represented by a $n\,{\times}\,n$ permutation matrix $\mathbf{X}\,{=}\,(x_{i,j})_{i,j=1,\ldots,n}$ with 
\begin{equation}
x_{i,j} = 
\begin{cases}
1 & \text{if } j =  \varphi(i) \\
0 & \text{else}
\end{cases}
\end{equation}
More generally, recall that a permutation matrix is defined as follows.
\begin{definition}[Permutation Matrix]\label{def:permutationmatrix}
A $n\,{\times}\,n$ matrix $\mathbf{X}$ is a permutation matrix iff it satisfies the following contraints:
\begin{equation}\label{eq-cts-mtx-perm}
  \left\lbrace\begin{aligned}
  &\forall j\,{=}\,1,\ldots,n,&&\sum_{i=1}^n x_{i,j} = 1 \\
  &\forall i\,{=}\,1,\ldots,n,&&\sum_{j=1}^n x_{i,j} = 1 \\
  &\forall i,j\,{=}\,1,\ldots,n,&&x_{i,j} \in \{0,1\}
  \end{aligned}\right.
  \end{equation}
\end{definition}
These constraints ensure $\mathbf{X}$ to be binary and doubly stochastic (sum of rows and sum of columns equal to $1$).

The selection of a relevant assignment, among all possible ones from $\mathcal{X}$ to $\mathcal{Y}$, depends on the problem. Nevertheless, each assignment is commonly penalized by a cost, and a relevant assignment becomes one having a minimal or a maximal cost. In this report, we consider minimal costs only. The cost of an assignment is usually defined as a sum of elementary costs. An elementary cost may penalize the assignment of an element of $\mathcal{X}$ to an element of $\mathcal{Y}$, or the simultaneaous assignment of two elements $i$ and $j$ of $\mathcal{X}$ to two elements $k$ and $l$ of $\mathcal{Y}$, respectively.

\begin{definition}[Linear Sum Assignment Problem (LSAP)]\label{def:lsap}
Let $\mathbf{C}\,{\in}\,[0,{+\infty})^{n\times n}$ be a matrix such that $c_{i,j}$ corresponds to the cost of assigning the element $i\,{\in}\,\mathcal{X}$ to the element $j\,{\in}\,\mathcal{Y}$. The Linear Sum Assignment Problem (LSAP) consists in finding
an optimal permutation 
\begin{equation}
\hat{\varphi}\in{\argmin}~\left\lbrace\sum_{i=1}^nc_{i,\varphi(i)}~|~\varphi\,{\in}\,\mathcal{S}_n\right\rbrace
\end{equation}
  where $\mathcal{S}_n$ is the set of all permutations of $\{1,\ldots,n\}$. Equivalently, the LSAP consists in finding an optimal $n\,{\times}\,n$ permutation matrix
\begin{equation}
\hat{\mathbf{X}}\in\argmin~\left\{\sum_{i=1}^n\sum_{j=1}^n c_{i,j} x_{i,j}~|~\mathbf{X}~\text{satisfies Eq.~\ref{eq-cts-mtx-perm}}\right\}.
\end{equation}
\end{definition}

Let $\mathbf{c}\,{=}\,\text{vec}(\mathbf{C})\,{\in}\,[0,{+\infty})^{n^2}$ be the vectorization of the cost matrix $\mathbf{C}$, obtained by concatenating its rows. Similarly, let $\mathbf{x}\,{=}\,\text{vec}(\mathbf{X})\,{\in}\,\{0,1\}^{n^2}$ be the vectorization of $\mathbf{X}$. Then, the LSAP consists in finding an optimal vector
\begin{equation}\label{eq:lsap_formulation}
  \hat{\mathbf{x}}\in\argmin\left\{\mathbf{c}^T\mathbf{x}~|~\mathbf{L}\mathbf{x}\,{=}\,\mathbf{1}_{2n},~\mathbf{x}\,{\in}\,\{0,1\}^{n^2}\right\},
\end{equation}
where the linear system $\mathbf{L}\mathbf{x}\,{=}\,\mathbf{1}$ is the matrix version of the constraints defined by Eq.~\ref{eq-cts-mtx-perm}. The matrix $\mathbf{L}\,{\in}\,\{0,1\}^{2n\times n^2}$ represents the node-edge incidence matrix of the complete bipartite graph $K_{n,n}$ with node sets $\mathcal{X}$ and $\mathcal{Y}$:
\begin{equation}
	l_{i,(j,k)}=\begin{cases}
	1&\text{if}~(j\,{=}\,i)\vee(k\,{=}\,i)\\
	0&\text{else}
	\end{cases}
\end{equation}
The system $\mathbf{L}\mathbf{x}\,{=}\,\mathbf{1}$, together with the binary constraint on $\mathbf{x}$, selects exactly one edge of $K_{n,n}$ for each element of $\mathcal{X}\,{\cup}\,\mathcal{Y}$. In other terms, these constraints represent a subgraph of $K_{n,n}$, with node sets $\mathcal{X}$ and $\mathcal{Y}$, such that each node has degree one. Indeed, the LSAP is a weighted bipartite graph matching problem.

More details on the LSAP can be found in \cite{sier01,bur09}. In particular, Eq.~\ref{eq:lsap_formulation} is a binary linear programming problem, efficiently solvable in polynomial time complexity, for instance with the Hungarian  algorithm~\cite{Kuhn1955,Munkres1957,law76} combined with pre-processing steps~\cite{Jonker1987}. In our experiments, we have used the $O(n^3)$ (time complexity) version of the Hungarian algorithm proposed in \cite{law76,bur09}.

Problems that incorporate pairwise constraints, \textit{i.e.} simultaneously assigning two elements of $\mathcal{X}$ to two elements of $\mathcal{Y}$, can be cast as a quadratic assignment problem \cite{koop57,Lowler1963,law76,Loiola2007,bur09}. This is the case for the graph matching problem, and for the GED as demonstrated in Section~\ref{sec:qap_formulation}. In this paper, we consider the general expression of quadratic assignment problems \cite{Lowler1963}.
\begin{definition}[Quadratic Assignment Problem (QAP)]\label{def:qsap}
Let $\mathbf{D}\in[0,{+\infty})^{n^2\times n^2}$ be a matrix such that $d_{ik,jl}$ corresponds to the cost of simultaneously assigning the elements $i$ and $j$ of $\mathcal{X}$ to the elements $k$ and $l$ of $\mathcal{Y}$, respectively. The quadratic assignment problem (QAP) consists in finding an optimal permutation
\begin{equation}\label{eq:lowler_form}
\hat{\varphi}\in\argmin\left\{\sum_{i=1}^n\sum_{j=1}^n d_{i\varphi(i),j\varphi(j)}~|~\varphi\in\mathcal{S}_n\right\}.
\end{equation}
Equivalently, the QAP consists in finding an optimal $n\,{\times}\,n$ permutation matrix
\begin{equation*}
\hat{\mathbf{X}}\in\argmin\left\{\sum_{i=1}^n\sum_{k=1}^n\sum_{j=1}^n\sum_{l=1}^n d_{ik,jl}x_{i,k}x_{j,l}~|~\mathbf{X}~\text{satisfies Eq.~\ref{eq-cts-mtx-perm}}\right\}.
\end{equation*}
\end{definition}
Note that $i\,{\in}\,\mathcal{X}$ is assigned to $k\,{\in}\,\mathcal{Y}$, and $j\,{\in}\,\mathcal{X}$ is assigned to $l\,{\in}\,\mathcal{Y}$, simultaneously iff $x_{i,k}\,{=}\,x_{j,l}\,{=}\,1$.

The QAP can be rewritten as a quadratic program:
\begin{equation*}
\argmin\left\{\mathbf{x}^T\mathbf{D}\mathbf{x}~|~\mathbf{L}\mathbf{x}\,{=}\,\mathbf{1}_{2n},~\mathbf{x}\,{\in}\,\{0,1\}^{n^2}\right\}
\end{equation*}
where $\mathbf{x}$ is the vectorization of $\mathbf{X}$, and the right-hand side is the matrix version of the constraints defined by Eq.~\ref{eq-cts-mtx-perm}.

The quadratic term is able to incorporate a linear one. Indeed, any simultenous assignment of the same element $i\,{\in}\,\mathcal{X}$ to the same element $k\,{\in}\,\mathcal{Y}$ is penalized by the cost $d_{ik,ik}$, \textit{i.e.} a diagonal element of $\mathbf{D}$. Since $x_{i,k}\,{\in}\,\{0,1\}$, we have $d_{ik,ik}x_{i,k}x_{i,k}\,{=}\,d_{ik,ik}x_{i,k}$. Then the total contribution of diagonal elements to the quadratic cost is given by
\begin{equation*}
	\sum_{i=1}^n\sum_{k=1}^n d_{ik,ik}x_{i,k}=\text{diag}(\mathbf{D})^T\mathbf{x}
\end{equation*}
where $\text{diag}$ denotes the diagonal vector. So, if $\text{diag}(\mathbf{D})\,{\not=}\,\mathbf{0}$, the quadratic functional incorporates a linear term which decribes the linear sum assignment between the elements of $\mathcal{X}$ and $\mathcal{Y}$. When these constraints are also part of the underlying problem, it is sometimes more convenient to rewritte the QAP as
\begin{equation}\label{eq:matrix_form}
\argmin\left\{\mathbf{x}^T\mathbf{D}\mathbf{x} + \mathbf{c}^T\mathbf{x}~|~\mathbf{L}\mathbf{x}\,{=}\,\mathbf{1}_{2n},~\mathbf{x}\,{\in}\,\{0,1\}^{n^2}\right\}
\end{equation}
where $\text{diag}(\mathbf{D})\,{=}\,\mathbf{0}$, and $\mathbf{c}$ is the cost of assigning each element of $\mathcal{X}$ to each element of $\mathcal{Y}$.

As the GED, the QAP is in general NP-hard, and exact algorithms can only be used with sets of small cardinality \cite{bur09}. Indeed, the cost functional is generally not convex, and methods based on relaxation and linearization are usually considered to find an approximate solution. See Section~\ref{sec:qap_solvers} for more details.

\section{Bipartite Graph Edit Distance}
\label{sec:linear_formulation}
As already mentionned, the GED can be challenging to compute, even on
small graphs.  In order to deal with such a complexity, several
methods approximate the GED by computing an optimal linear sum
assignment between the nodes of the two graphs to be compared. This is
formalized through the definition of a specific
LSAP~\cite{Riesen2007,Riesen2009} that takes into account edit
operations, as described in this section. In particular, we show in
Section~\ref{sec:line-sum-assignm} that there is a one-to-one relation
between restricted edit paths and  assignment matrices. 
\subsection{Edit operations, assignments and restricted edit paths}\label{sec:line-sum-assignm}
Let $\mathcal{X}$ and $\mathcal{Y}$ be two finite sets, with $|\mathcal{X}|=n$ and $|\mathcal{Y}|=m$. Without loss of generality, we assume that $\mathcal{X}\,{=}\,\{1,\ldots,n\}$ and $\mathcal{Y}\,{=}\,\{1,\ldots,m\}$. Each element of $\mathcal{X}$ can be assigned to an element of $\mathcal{Y}$. Such a mapping represents a possible substitution. Also each element of $\mathcal{X}$ can be removed, and each element of $\mathcal{Y}$ can be inserted into $\mathcal{X}$. In order to represent insertions, $\mathcal{X}$ is augmented by $m$ dummy elements $\mathcal{E}_\mathcal{X}\,{=}\,\{\epsilon_1,\ldots,\epsilon_m\}$, such that $j\,{\in}\,\mathcal{Y}$ can only be inserted into $\mathcal{Y}$ by assiging $\epsilon_j$ to $j$. Similarly, the set $\mathcal{Y}$ is augmented by $n$ dummy elements $\mathcal{E}_\mathcal{Y}\,{=}\,\{\epsilon_1,\ldots,\epsilon_n\}$, such that $i\,{\in}\,\mathcal{X}$ is removed by assigning it to $\epsilon_i$. In other terms, it is not possible to assign an element $i\,{\in}\,\mathcal{X}$ to an element $\epsilon_k\,{\in}\,\mathcal{E}_\mathcal{Y}$ with $k\,{\not=}\,i$, and similarly any assignment from $\epsilon_j\,{\in}\,\mathcal{E}_\mathcal{X}$ to $k\,{\in}\,\mathcal{Y}$ with $k\,{\not=}\,j$ is forbidden. 

Let $\mathcal{X}^\epsilon\,{=}\,\mathcal{X}\,{\cup}\,\mathcal{E}_\mathcal{X}$ and $\mathcal{Y}^\epsilon\,{=}\,\mathcal{Y}\,{\cup}\,\mathcal{E}_\mathcal{Y}$ be the two augmented sets, which thus have the same size $n\,{+}\,m$. We
assume without loss of generality that symbols $\epsilon_i$ and $\epsilon_j$ represent integers, \textit{i.e.} $\mathcal{E}_\mathcal{X}\,{=}\,\{n+1,\ldots,n+m\}$ and $\mathcal{E}_\mathcal{Y}\,{=}\,\{m+1,\ldots,m+n\}$. It is now possible to define assignments that take into account removal, substitution, and insertion of elements.
\begin{definition}[$\epsilon$-assignment]\label{def-epsass}
	An $\epsilon$-assignment from $\mathcal{X}$ to $\mathcal{Y}$ is a bijective mapping $\psi\,{:}\,\mathcal{X}^\epsilon\,{\rightarrow}\,\mathcal{Y}^\epsilon$, here a permutation, such that for each element of $\mathcal{X}^\epsilon$ one of the four following cases occurs:
\begin{enumerate}
\item Substitutions: $\psi(i)\,{=}\,j$ with $(i,j)\,{\in}\,\mathcal{X}\,{\times}\,\mathcal{Y}$.
\item Removals: $\psi(i)\,{=}\,\epsilon_i$ with $i\,{\in}\,\mathcal{X}$.
\item Insertions: $\psi(\epsilon_j)\,{=}\,j$ with $j\,{\in}\,\mathcal{Y}$.
\item Finally $\psi(\epsilon_j)\,{=}\,\epsilon_i$ allow to complete 
  the bijective property of $\psi$, and thus should be ignored. This occurs when $i\,{\in}\,\mathcal{X}$ and $j\,{\in}\,\mathcal{Y}$ are both substituted.
\end{enumerate}
Let $\Psi_\epsilon(\mathcal{X},\mathcal{Y})$ be the set of all $\epsilon$-assignments from $\mathcal{X}$ to $\mathcal{Y}$.
\end{definition} 
In other terms, an $\epsilon$-assignment is a bijection (or permutation) with additional constraints. The corresponding $(n\,{+}\,m)\,{\times}\,(m\,{+}\,n)$ permutation matrix can be decomposed as follows:
\begin{equation}\label{def:assMatrix}
	\mathbf{X} ={~} \begin{blockarray}{@{\hspace{2pt}}c|c@{\hspace{2pt}}cc}
	1\cdots m&\epsilon_1\cdots\epsilon_n &&\\
	\begin{block}{(@{\hspace{2pt}}c|c)@{\hspace{2pt}}cc}\vspace{-0.15cm}
		&&&1\\\vspace{-0.15cm}
		\mathbf{X}^\text{sub}& \mathbf{X}^\text{rem} &&\vdots\\
		&  &&n\\\cline{1-4}\vspace{-0.15cm}
		& &&\epsilon_1\\\vspace{-0.15cm}
		\mathbf{X}^\text{ins}&\mathbf{X}^{\epsilon}&&\vdots\\\vspace{0.15cm}
		&  &&\epsilon_m\\
	\end{block}
	\end{blockarray}
\end{equation}
where matrix $\mathbf{X}^{\text{sub}}\,{\in}\,\{0,1\}^{n\times m}$ encodes node substitutions, $\mathbf{X}^{\text{rem}}\,{\in}\,\{0,1\}^{n\times n}$ encodes node removals, and $\mathbf{X}^{\text{ins}}\,{\in}\,\{0,1\}^{m\times m}$ encodes node insertions. Matrix $\mathbf{X}^\epsilon\,{\in}\,\{0,1\}^{m\times n}$ is an auxiliary matrix (case 4 above), it 
 ensures that $\mathbf{X}$ is a permutation matrix. Due to the constraints on dummy nodes 
 (cases 2 and 3 above) matrices $\mathbf{X}^{\text{rem}}$ and 
 $\mathbf{X}^{\text{ins}}$ always satisfy:
\begin{equation}\label{eq-cst-eps}
\begin{aligned}
&\forall (i,j)\in \{1,\dots,n\}^2, i\not=j,&& x_{i,j}^\text{rem}\,{=}\,0\\
&\forall (i,j)\in \{1,\dots,m\}^2, i\not=j,&& x_{i,j}^\text{ins}\,{=}\,0.
\end{aligned}
\end{equation}
\begin{definition}[$\epsilon$-assignment matrix]\label{def:epsmat}
	 A $(n\,{+}\,m)\,{\times}\,(m\,{+}\,n)$ matrix satisfying equations \ref{eq-cts-mtx-perm}, \ref{def:assMatrix} and \ref{eq-cst-eps} is called an $\epsilon$-assignment matrix. The set of all $(n\,{+}\,m)\,{\times}\,(m\,{+}\,n)$ $\epsilon$-assignment matrices is denoted by $\mathcal{A}_{n,m}$.
\end{definition}
The auxiliary matrix $\mathbf{X}^\epsilon$ in Eq.~\ref{def:assMatrix} suggests the definition of an equivalence relation between $\epsilon$-assignment matrices.
\begin{definition}\label{def:equiv_Am}
  Two $\epsilon$-assignment matrices $\mathbf{X}_1$ and $\mathbf{X}_2$, defined by the two sequences
  of block matrices $(\mathbf{X}^{\text{sub}}_1,\mathbf{X}^{\text{rem}}_1,\mathbf{X}^{\text{ins}}_1,\mathbf{X}^\epsilon_1)$ and
  $(\mathbf{X}^{\text{sub}}_2,\mathbf{X}^{\text{rem}}_2,\mathbf{X}^{\text{ins}}_2,\mathbf{X}^\epsilon_2)$, are equivalent iff $$(\mathbf{X}^{\text{sub}}_1=\mathbf{X}^{\text{sub}}_2)\wedge(\mathbf{X}^{\text{rem}}_1=\mathbf{X}^{\text{rem}}_2)\wedge(\mathbf{X}^{\text{ins}}_1=\mathbf{X}^{\text{ins}}_2).$$
The set of $\epsilon$-assignment matrices up to this equivalence relation is denoted by $\mathcal{A}_{n,m}^\sim$.
\end{definition}
\begin{proposition}\label{prop:corresp_AmMappings}
  There is a one-to-one relation between $\mathcal{A}_{n,m}^\sim$
  and the set of injective functions from a subset of $\mathcal{X}$ to
  $\mathcal{Y}$.
\end{proposition}
\begin{proof}
  Recall that $\mathcal{X}\,{=}\,\{1,\ldots,n\}$ and $\mathcal{Y}\,{=}\,\{1,\ldots,m\}$. 
  Let $\mathbf{X}=(\mathbf{X}^\text{sub},\mathbf{X}^\text{rem},\mathbf{X}^\text{ins},\mathbf{X}^\epsilon)=(\mathbf{Q},\mathbf{R},\mathbf{S},\mathbf{T})$ denote an $\epsilon$-assignment matrix. If $\mathbf{Q}\,{=}\,\mathbf{0}$ we associate
  to $\mathbf{X}$ the application $\varphi_0$ from the empty set onto itself (Proposition~\ref{prop:edit_path_mapping}). Otherwise, let
  us introduce the set:
  \[
  \hat{\mathcal{X}}=\{i\,{\in}\,\mathcal{X}~|~\exists j\,{\in}\,\mathcal{Y},~q_{i,j}\,{=}\,1\}
  \]
  Since $\mathbf{X}$ is a permutation matrix, for any $i\,{\in}\,\hat{\mathcal{X}}$
  there is exactly one $j\,{\in}\,\mathcal{Y}$ such that $q_{i,j}\,{=}\,1$. We can thus define
  the mapping:
  \[
  \varphi_\mathbf{X} \,\left(
    \begin{array}{lll}
      \hat{\mathcal{X}}      &\rightarrow & \mathcal{Y}\\
      i&\mapsto&j\\
    \end{array}
    \right)
    \]
    Moreover if $\varphi_\mathbf{X}(i_1)\,{=}\,\varphi_\mathbf{X}({i_2})$ then we have
    $q_{i_1,j}\,{=}\,q_{i_2,j}\,{=}\,1$. Since $\mathbf{X}$ is a permutation matrix, such a
    case is possible only if $i_1\,{=}\,i_2$, and $\varphi_\mathbf{X}$ is thus injective. We can thus associate to each assignement matrix an injective
    function $\varphi_\mathbf{X}$ from a subset of $\mathcal{X}$ to
    $\mathcal{Y}$. We denote by $\chi$ this mapping. We have to show
    that $\chi$ is bijective.

    Consider an injective mapping $\varphi$ from a subset
    $\hat{\mathcal{X}}$ onto $\mathcal{Y}$. If $\varphi\,{=}\,\varphi_0$, we have $\chi(\mathbf{X})\,{=}\,\varphi_0$ with
    $\mathbf{X}\,{=}\,(\mathbf{0}_{n\times m},\mathbf{I}_{n\times n},\mathbf{I}_{m\times m},\mathbf{0}_{m\times n})$.  Otherwise, the sub-blocks of $\mathbf{X}$ are defined as follows:
    \[
\begin{aligned}
&q_{i,j}=\begin{cases}
	1&\mbox{iff}~\exists i\in \hat{\mathcal{X}},\mbox{ with  } \varphi(i)=j,\\
    0&\mbox{else}
    \end{cases}\\
&r_{i,i}=\begin{cases}
	1&\mbox{iff}~i\in \mathcal{X}\setminus\hat{\mathcal{X}}\\
    0&\mbox{else}
    \end{cases}\\
&s_{j,j}=\begin{cases}
	1&\mbox{iff}~j\in \mathcal{Y}\setminus\varphi[\hat{\mathcal{X}}]\\
    0&\mbox{else}
    \end{cases}
\end{aligned}
    \]
    Note that off-diagonal elements of $\mathbf{R}$ and $\mathbf{S}$ are equal to $0$ by
    definition of $\epsilon$-assignement matrices (Eq.~\ref{eq-cst-eps}).
    
    By using the above definition, $q_{i,j}\,{=}\,1\Rightarrow r_{i,i}\,{=}\,s_{j,j}\,{=}\,0$.
    Hence, if $\mathbf{T}$ was filled with $0$, the line corresponding to
    $\epsilon_i\,{\in}\,\mathcal{Y}^\epsilon$ (column $m\,{+}\,i$) and
    $\epsilon_j\,{\in}\,\mathcal{X}^\epsilon$ (row $n\,{+}\,j$) would be filled by $0$ in $\mathbf{X}$, which would thus not be an $\epsilon$-assignment matrix. Moreover, since $\varphi$ is injective, we have $|\hat{\mathcal{X}}|\,{=}\,|\varphi[\hat{\mathcal{X}}]|$. Hence the set
    of indices:
    \[
    \mathcal{A}=\{(j,i)\in\{1,\dots,m\}\times\{1,\dots,n\}~|~q_{i,j}\,{=}\,1\}
    \]
    defines a square submatrix of $\mathbf{T}$, which can be defined as
    a permutation matrix on $\mathcal{A}$ and $0$ elsewhere. Now we check that each row and column of $\mathbf{X}=(\mathbf{Q},\mathbf{R},\mathbf{S},\mathbf{T})$ contains exactly one value equal to $1$.
    \paragraph*{For the rows:}
      \begin{itemize}
      \item If $i\in \hat{\mathcal{X}}$, then $\varphi(i)=j$, $q_{i,j}=1$ and
        $r_{i,i}=0$. Moreover, since $\varphi$ is an application, we have
        $q_{i,j'}=0$ for any $j'\in\mathcal{Y}\setminus\{j\}$. 
      \item If $i\in \mathcal{X}\setminus\hat{\mathcal{X}}$, by definition $q_{i,j}=0$ for
        all $j\in\mathcal{Y}$ and $r_{i,i}=1$. There is thus a single
        $1$ in row $i$.
      \item For $j\in\mathcal{Y}$ with $s_{j,j}=1$, by definition of $\mathbf{S}$ we have
        $j\in\mathcal{Y}\setminus\varphi[\mathcal{\hat{X}}]$, and thus 
        $\forall i\in \mathcal{X},~q_{i,j}=0$. Hence there is no $i$ such that
        $(j,i)\in \mathcal{A}$. By definition of $\mathbf{T}$, its row
        $j$ is filled with $0$.
      \item For $j\in\mathcal{Y}$ with $s_{j,j}=0$, then
        $j\in \varphi[\mathcal{X}]$ and it exists a unique
        $i\in \mathcal{X}$ such that $q_{i,j}=1$. Hence
        $(j,i)\in \mathcal{A}$. Since $\mathbf{T}$ defines a permutation matrix
        on $\mathcal{A}$, it should exists a unique $k$ satisfying
        $(j,k)\in\mathcal{A}$ such that $t_{j,k}=1$. This value is
        unique on the row $j$ of $\mathbf{T}$ by definition of $\mathbf{T}$.
      \end{itemize}
      Hence for each element of $\mathcal{X}^\epsilon$, there is a unique $1$ on the corresponding row of $\mathbf{X}$. 
      
      The proof for columns is similar.
%
%
	So there is a unique $1$ value on each row and each column
    of $\mathbf{X}$. Matrix $\mathbf{X}$ is thus a permutation matrix. Moreover, since elements of the blocks $\mathbf{R}$ and $\mathbf{S}$ are equal to $0$ on off-diagonal indices, $\mathbf{X}$ is an $\epsilon$-assignement matrix with $\chi(\mathbf{X})=\varphi$, by construction.

    The application $\chi$ is thus surjective. Let us show that it is
    also injective. To this end, we consider two non-equivalent
    assignement matrices $\mathbf{X}=(\mathbf{Q},\mathbf{R},\mathbf{S},\mathbf{T})$ and $\mathbf{X}'=(\mathbf{Q}',\mathbf{R}',\mathbf{S}',\mathbf{T}')$.
    \begin{itemize}
    \item If $\mathbf{Q}\neq\mathbf{Q}'$, then we may suppose without loss of
      generality that there exists
      $(i,j)\in \mathcal{X}\times\mathcal{Y}$ such that
      $q_{i,j}=1$ and $q'_{i,j}=0$. Since $q_{i,j}=1$,
      $i\in \mathcal{\hat{X}}_\mathbf{X}$, where $\mathcal{\hat{X}}_\mathbf{X}$
      denotes the subset of $\mathcal{X}$ on which $\chi(\mathbf{X})$ is
      defined.
      \begin{itemize}
      \item If $q'_{i,j'}=0$ for all $j'\in\mathcal{Y}$ then
        $i\in \mathcal{X}\setminus\mathcal{\hat{X}}_{\mathbf{X}'}$.  Hence $\chi(\mathbf{X})$ and
        $\chi(\mathbf{X}')$ are not defined on the same set and are
        consequently not equal.
      \item If it exists $j'\in\mathcal{Y}$ such that $q'_{i,j'}=1$
        we have $\chi(\mathbf{X})(i)=j$ and
        $\chi(\mathbf{X}')(i)=j'$. Applications $\chi(\mathbf{X})$ and $\chi(\mathbf{X}')$
        are consequently not equal.
      \end{itemize}
    \item If $\mathbf{R}\neq\mathbf{R}'$ we may suppose that it exists
      $i\in\mathcal{X}$ such that $r_{i,i}=1$ and $r'_{i,i}=0$. Hence
      $i\not\in \mathcal{\hat{X}}_\mathbf{X}$ and
      $i\in \mathcal{\hat{X}}_{\mathbf{X}'}$. Applications $\chi(\mathbf{X})$ and
        $\chi(\mathbf{X}')$ are not defined on the same set and are
        consequently not equal.
      \item If $\mathbf{S}\neq\mathbf{S}'$, let us consider $j\in\mathcal{Y}$ such
        that $s_{j,j}=1$ and $s'_{j,j}=0$. We have
        $j\not\in \chi(\mathbf{X})[\mathcal{\hat{X}}_\mathbf{X}]$ and
        $j\in \chi(\mathbf{X}')[\mathcal{\hat{X}}_{\mathbf{X}'}]$. Applications
        $\chi(\mathbf{X})$ and $\chi(\mathbf{X}')$ have different set of mappings and
        are consequently not equal.
    \end{itemize}
    Note that we do need to consider matrix $\mathbf{T}$ since this matrix is
    not implied in the equivalence relationship. In all cases
    application $\chi$ maps two non-equivalent $\epsilon$-assignement matrices to
    different injective functions. The application $\chi$ is thus injective. 

    Application $\chi$ being both surjective and injective, it is bijective.
\end{proof}

~\\\noindent
So, $\epsilon$-assignment matrices on $\mathcal{X}^\epsilon$ and $\mathcal{Y}^\epsilon$, and
injective functions defined on a subset of $\mathcal{X}$ onto $\mathcal{Y}$, are in one-to-one correspondence. 

It is now possible to link $\epsilon$-assignments to edit paths. Consider two simple graphs  $G_1$ and $G_2$, with node sets $V_1$ and $V_2$ respectively. An $\epsilon$-assignment from $V_1$ to $V_2$ can be defined by constructing the sets $V_1^\epsilon$ and $V_2^\epsilon$. According to the above proposition, there is a one-to-one correspondence between $\epsilon$-assignment matrices on $V_1^\epsilon$ and $V_2^\epsilon$, and injective functions defined on a subset of $V_1$ onto $V_2$. By using Proposition~\ref{prop:edit_path_mapping}, we can connect such a mapping to a restricted edit path between $G_1$ and $G_2$ (Definition \ref{def:editPath}). Up to the equivalence relation (Definition~\ref{def:equiv_Am}), there is thus a one-to-one correspondence between $\epsilon$-assignment matrices and restricted edit paths.

This shows that restricted edit paths can be deduced from $\epsilon$-assignments.
\subsection{LSAP for $\epsilon$-assignments}\label{sec-lsapeps}
Let $\mathcal{X}\,{=}\,\{1,\ldots,n\}$ and $\mathcal{Y}\,{=}\,\{1,\ldots,m\}$ be two sets. These two sets are augmented by dummy elements as described in the previous section, \textit{i.e.} $\mathcal{X}^\epsilon\,{=}\,\mathcal{X}\,{\cup}\,\mathcal{E}_\mathcal{X}$ and $\mathcal{Y}^\epsilon\,{=}\,\mathcal{Y}\,{\cup}\,\mathcal{E}_\mathcal{Y}$. An $\epsilon$-assignment from $\mathcal{X}$ to $\mathcal{Y}$, \textit{i.e.} a bijective mapping from $\mathcal{X}^\epsilon$ onto $\mathcal{Y}^\epsilon$, represents a set of edit operations.

The selection of a relevant $\epsilon$-assignment is realized through the design of a pairwise cost function adapted to edit operations. To this, each possible mapping of an element $i\,{\in}\,\mathcal{X}^\epsilon$ to an element $j\,{\in}\,\mathcal{Y}^\epsilon$ is penalized by a non-negative cost $c_{i,j}$. All costs can be encoded by a $(n\,{+}\,m)\,{\times}\,(m\,{+}\,n)$ matrix (having the same structure as $\epsilon$-assignment matrices)  \cite{Riesen2007,Riesen2009}
\begin{equation}\label{eq:cost-matrix}
	\mathbf{C} = \begin{blockarray}{@{\hspace{2pt}}c|c@{\hspace{2pt}}cc}
	1\cdots m&\epsilon_1\cdots\epsilon_n &&\\
	\begin{block}{(@{\hspace{2pt}}c|c)@{\hspace{2pt}}cc}\vspace{-0.15cm}
		&&&1\\\vspace{-0.15cm}
		\mathbf{C}^\text{sub}& \mathbf{C}^\text{rem} &&\vdots\\
		&  &&n\\\cline{1-4}\vspace{-0.15cm}
		& &&\epsilon_1\\\vspace{-0.15cm}
		\mathbf{C}^\text{ins}&\mathbf{0}_{m,n}&&\vdots\\\vspace{0.15cm}
		&  &&\epsilon_m\\
	\end{block}
	\end{blockarray}
\end{equation}
where the matrix
$\mathbf{C}^\text{sub}\,{\in}\,[0,{+\infty})^{n\times m}$ encodes
substitution costs,
$\mathbf{C}^\text{rem}\,{\in}\,[0,{+\infty})^{n\times n}$ encodes
removal costs, and
$\mathbf{C}^\text{ins}\,{\in}\,[0,{+\infty})^{m\times m}$ encodes
insertion costs. According to cases 2 and 3 in Definition~\ref{def-epsass}, off-diagonal values
of $\mathbf{C}^\text{rem}$ and $\mathbf{C}^\text{ins}$ are typically
set to a large value $\omega$, such that
$\max\{c_{i,\psi(i)}\,|\,{\forall i\,{\in}\,\mathcal{X}^\epsilon},
{\forall\psi\,{\in}\,\Psi_\epsilon(\mathcal{X},\mathcal{Y})}\}\,{\ll}\,\omega\,{<}\,{+\infty}$,
in order to avoid forbidden mappings. According to case 4 in Definition~\ref{def-epsass}, the
mapping of any $\epsilon_i$ to an $\epsilon_j$ should not induce any
cost, so the last block of $\mathbf{C}$ is set to the null matrix
$\mathbf{0}_{m,n}$. The cost of an $\epsilon$-assignment $\psi$ can then be measured by the sum (see Definition~\ref{def-epsass} for the decomposition)
\begin{equation}
\label{eq-lsappsi}
	\sum_{i=1}^{n+m}c_{i,\psi(i)}=\underset{\text{substitutions}}{\underbrace{\sum_{i\in\hat{\mathcal{X}}}c_{i,\psi(i)}}}+\underset{\text{removals}}{\underbrace{\sum_{i\in\mathcal{X}\setminus\hat{\mathcal{X}}}c_{i,\epsilon_i}}}+\underset{\text{insertions}}{\underbrace{\sum_{j\in \mathcal{Y}-\psi[\mathcal{\hat{X}}]}c_{\epsilon_j,j}}}.
\end{equation}
where $\hat{\mathcal{X}}\,{=}\,\{i\,{\in}\,\mathcal{X}~|~{\exists j\,{\in}\,\mathcal{Y}},~\psi(i)\,{=}\,j\}$.

An optimal $\epsilon$-assignment is then defined as one having a minimal cost (several optimal $\epsilon$-assignment may exist) among all $\epsilon$-assignments:
\begin{equation}\label{pb-lsap-psi}
\hat{\psi}\in\argmin\left\{\sum_{i=1}^{n+m}c_{i,\psi(i)}~|~\psi\,{\in}\,\Psi_\epsilon(\mathcal{X},\mathcal{Y})\right\}
\end{equation}
which is a LSAP (Section~\ref{sec:mainDefAssignment}). It can thus be rewritten as a binary programming problem (Eq.~\ref{eq:lsap_formulation})
\begin{equation}\label{eq-lsapeps}
	\hat{\mathbf{x}}\in\argmin\left\{\mathbf{c}^T\mathbf{x}~|~\mathbf{x}\,{\in}\,\text{vec}[\mathcal{A}^\sim_{n,m}]\right\},
\end{equation}
where
$\mathbf{x}\,{=}\,\text{vec}(\mathbf{X})\,{\in}\,\{0,1\}^{(n+m)^2}$ is
the vectorization of the $\epsilon$-assignment matrix $\mathbf{X}$
associated with $\psi$ (Eq.~\ref{def:assMatrix}),
$\mathbf{c}\,{=}\,\text{vec}(\mathbf{C})\,{\in}\,[0,{+\infty})^{(n+m)^2}$
is the vectorization of the cost matrix $\mathbf{C}$, and
$\text{vec}[\mathcal{A}^\sim_{n,m}]\,{\subset}\,\{0,1\}^{(n+m)^2}$ is
the set of all vectorized $\epsilon$-assignment matrices. Note that,
using equation~\ref{eq-lsappsi}, two equivalent $\epsilon$-assignment
matrices have a same cost.

The optimal solution of the LSAP defined by Eq.~\ref{eq-lsapeps} can be computed by any algorithm that solves LSAPs, such as Hungarian-type algorithms. Note that mappings in $\Psi_\epsilon$, or matrices in $\mathcal{A}_{n,m}$, are much more constrained than bijective mappings or permutation matrices. These constraints, \textit{i.e.} forbidden assignments, are satisfied in \cite{Riesen2007,Riesen2009} through the large $\omega$ values in the cost matrix. This is a classical trick used in LSAPs to avoid some specific assignments of elements \cite{bur09}. While these assignments are avoided, the corresponding large $\omega$ values are still treated by the algorithms. A better way to take into account the additionnal constraints would be to modify the algorithms such that forbidden assignments are not treated at all. This is the choice we made in our experimentations. This improves the time complexity.
\subsection{Bipartite GED}\label{sec:geda}
It is now possible to define a framework to approximate the GED, based
on $\epsilon$-assignments and the corresponding LSAP
\cite{Riesen2007,Zeng2009,Riesen2009,Gauzere2014a,rie15}. Within this
framework, a resulting approximate GED is called a bipartite GED.

Let $G_1$ and $G_2$ be two graphs, with node sets $V_1$ and $V_2$ respectively. The computation of a bipartite GED from $G_1$ to $G_2$ is performed in four main steps detailed below.
\paragraph*{Step 1 - construction of the bags of patterns.}For each
node of each graph a bag of patterns is constructed. This bag
represents a part of the graph connected to a specific node by some
structured subgraphs, such as incident edges
\cite{Riesen2007,Riesen2009}, subtrees~\cite{Zeng2009} or walks
\cite{Gauzere2014a}. A set of bags of patterns is then obtained for
each graph. Let $\mathcal{B}_1$ and $\mathcal{B}_2$ be the ones
associated to $G_1$ and to $G_2$ respectively. We have thus
$|\mathcal{B}_1|\,{=}\,|V_1|$ and $|\mathcal{B}_2|\,{=}\,|V_2|$.  The
idea is then to find an optimal $\epsilon$-assignment from
$\mathcal{B}_1$ to $\mathcal{B}_2$, according to a given pairwise cost
matrix.
\paragraph*{Step 2 - construction of the cost matrix.}Each possible mapping of a bag $b_i\,{\in}\,\mathcal{B}_1$ to a bag $b_j\,{\in}\,\mathcal{B}_2$ is penalized by a cost measuring the affinity between the two bags. This cost is initially defined as the cost of editing $b_i$ such that it is transformed into $b_j$, \textit{i.e.} the cost of an optimal $\epsilon$-assignment of the elements of the two bags \cite{Riesen2007,Riesen2009,Gauzere2014a,rie15}. Also the bags of $\mathcal{B}_1$ can be removed, and the bags of $\mathcal{B}_2$ can be inserted into $\mathcal{B}_1$, which is penalized by a cost measuring  the relevance of the bag. In order to approximate the GED, all theses costs depend on the original edit cost functions defined in Section~\ref{sec:mainDefEditDist}. They are encoded by a cost matrix $\mathbf{C}$ (Eq.~\ref{eq:cost-matrix}). Note that in this step $|V_1|\,{\times}\,|V_2|$ LSAP are solved for computing the costs of assigning bags of $\mathcal{B}_1$ to bags of $\mathcal{B}_2$.
\paragraph*{Step 3 - construction of an $\epsilon$-assignment between the nodes.} Given the cost matrix $\mathbf{C}$ computed in the previous step, an optimal $\epsilon$-assignment from $\mathcal{B}_1$ to $\mathcal{B}_2$ is then computed by solving again a LSAP. The computed optimal assignment hence provides an optimal
mapping $\psi\,{\in}\,\Psi_\epsilon(V_1,V_2)$.
\paragraph*{Step 4 - construction of a restricted edit path.} The $\epsilon$-assignment $\psi$ can be interpreted as
a partial edit path between the graphs $G_1$ and $G_2$. Indeed, it is only composed of edit
operations involving nodes. Therefore this partial edit path has to be
completed with edit operations applied on edges. This set of edit
operations is directly induced by the set of edit operations operating
on nodes, defined by the mapping computed in the previous step. The
substitution, removal or insertion of any edge depends thus on the
edit operations performed on its incident nodes. The cost of the complete edit
path is finally defined by the sum of edit operations on nodes and
edges. This cost only corresponds to an approximation of the GED between
$G_1$ and $G_2$ since the mapping computed during Step 3 may not be
optimal. Therefore, this cost corresponds to an overestimation of the
exact GED, known as bipartite GED.

~\\\indent
The definition of the cost matrix $\mathbf{C}$ in Step 2 is a keypoint of the framework.
%
%
The initial approach proposed in~\cite{Riesen2007,Riesen2009} defines bag of
patterns as the corresponding node itself and its direct
neighborhood, \text{i.e.} the set of incident edges and adjacent
nodes. The cost of assigning a bag $b_i\,{\in}\,\mathcal{B}_1$ to a bag $b_j\,{\in}\,\mathcal{B}_2$ is then defined as the substitution
cost of the associated node $i\,{\in}\,V_1$ and $j\,{\in}\,V_2$, plus the cost associated to an optimal $\epsilon$-assignment between the two sets composed of their incident edges and their adjacent nodes. Using such bags of patterns can be discriminant enough, in which case the bipartite GED provides a good approximation of the GED. But this approach lacks of
accuracy in some cases, in particular when the direct neighbourhood of the nodes
is homogeneous in the graph. When considering such graphs, the
cost associated to each pair of bags do not differ
sufficiently, and the optimal $\epsilon$-assignment depends more on the traversal of the nodes by the LSAP solver than on the graph's structure.

In order to improve the accuracy of the bipartite GED, the information attached to each node needs to be more global, for instance by considering bags of
walks up to a length $k$ \cite{Gauzere2014a}, instead of the direct
neighbourhood. This approach follows the same scheme as the one used
in \cite{Riesen2007,Riesen2009}, except that patterns 
associated to a node are defined as walks of length $k$ starting at this
node. Considering bags of such
patterns allow to extend the amount of information associated to the
nodes, which leads to a better approximation of the GED. However, the use of
bags of walks induces some drawbacks. First, the set of computed walks
suffers from the tottering phenomenon which leads to consider
irrelevant patterns, especially when considering high values of
$k$. These irrelevant patterns affect the cost of the $\epsilon$-assignment, and thus the
quality of the approximation of the GED. In addition, the mapping
cost between two bags of walks can only be approximated, which induces
another loss of accuracy.

These drawbacks can be avoided by using bags of subgraphs
rather than bags of walks, such as all subgraphs centered on a given
node and up to a radius $k$ \cite{Carletti2015}. The cost associated to the mapping of two
bags of such patterns is defined as the edit distance between the
two $\kgraphs$ centered on the respective nodes. Despite the fact that we can
control the size of these subgraphs thanks to the parameter $k$, this
approach requires significantly more computational time than the
previous ones. However, the use of accurate sub-structures allows to
obtain a better approximation of the GED.


\section{GED as a Quadratic Assignment Problem}
\label{sec:qap_formulation}
The bipartite GED is a good candidate approximation of the GED, but it
is based on the construction of a restricted edit path which generally
does not have a minimal cost. Costs on edges can only be deduced from
operations performed on their two incident nodes. This cannot be taken
into account by the approach based on the LSAP, which considers
information about edges separately in each node. To fully describe the GED, the model must take into account simultaneous node and edge assignments. This can be formalized as a quadratic assignment problem \cite{Riesen2007,Riesen2009}. In this section, we propose to extend the linear model to a quadratic one based on $\epsilon$-assignments, and we show that this model corresponds to the cost of a restricted edit path.
\subsection{Simultaneous node assignment and quadratic cost}
Let $G_1\,{=}\,(V_1,E_1)$ and $G_2\,{=}\,(V_2,E_2)$ be two graphs, and 
let $\psi\in \Psi(V_1^\epsilon,V_2^\epsilon)$ be an $\epsilon$-assignment (Definition~\ref{def-epsass}). When a pair $(i,j)\in V_1^\epsilon\times V_1^\epsilon$ is assigned by $\psi$ to a pair $(\psi(i),\psi(j))\in V_2^\epsilon\times V_2^\epsilon$, one of the following cases occurs:
\begin{enumerate}
\item Edge substitution: $(\psi(i),\psi(j))\,{\in}\,E_2$ with $(i,j)\,{\in}\,E_1$.
\item Edge removal: $(\psi(i),\psi(j))\,{\not\in}\,E_2$ with $(i,j)\,{\in}\,E_1$.
\item Edge insertion: $(\psi(i),\psi(j))\,{\in}\,E_2$ with $(i,j)\,{\not\in}\,E_1$.
\item Finally $(\psi(i),\psi(j))\,{\not\in}\,E_2$ with $(i,j)\,{\not\in}\,E_1$ allows to complete the bijection property.
\end{enumerate}
Each possible simultaneous mapping of nodes $i,j\,{\in}\,V_1^\epsilon$
onto respectively nodes $k$ and $l$ in $V_2^\epsilon$, is penalized by
a non-negative cost $d_{ik,jl}$ which depends on the underlying edit
operation described by one of the cases above. The overall edge's cost
associated to a simultaneous node assignment is then measured by:
\begin{equation}
	d(\psi)=\sum_{i=1}^{n+m}\sum_{j=1}^{n+m}d_{i\psi(i),j\psi(j)},
\end{equation}
where cost values are defined as follows.

Recall that all mappings from a node of $V_1^\epsilon$ to a node of $V_2^\epsilon$ are not allowed. Indeed (Section \ref{sec:line-sum-assignm}), $i\,{\rightarrow}\,\epsilon_j$ with $i\,{\in}\,V_1$ and $j\,{\not=}\,i$, and reciprocally $\epsilon_k\,{\rightarrow}\,l$ with $l\,{\in}\,V_2$ and $k\,{\not=}\,l$ are forbidden. Then, a simultaneous node mapping involving at least one of these two cases is also forbidden. We denote by $\not\rightarrow$ a forbidden mapping. As in Section \ref{sec-lsapeps}, the cost is set to a (large) value $\omega$ in this case.

For any other simultaneous node mapping $(i\,{\rightarrow}\,k,j\,{\rightarrow}\,l)$, with $i,j\,{\in}\,V_1^\epsilon$ and $k,l\,{\in}\,V_2^\epsilon$, its cost depends on the presence or the absence of edges $(i,j)\,{\in}\,E_1$ and $(k,l)\,{\in}\,E_2$:
\begin{enumerate}
	\item[$\bullet$]If $(i,j)\,{\in}\,E_1$ and $(k,l)\,{\in}\,E_2$ then $d_{ik,jl}$ is the cost of the edge assignment $(i,j)\,{\rightarrow}\,(k,l)$, \textit{i.e.} edge substitution.
	\item[$\bullet$] If $(i,j)\,{\in}\,E_1$ and $(k,l)\,{\not\in}\,E_2$ then $d_{ik,jl}$ is the cost of removing the edge $(i,j)$.
	\item[$\bullet$] If $(i,j)\,{\not\in}\,E_1$ and $(k,l)\,{\in}\,E_2$ then $d_{ik,jl}$ is the cost of inserting the edge $(k,l)$.
	\item[$\bullet$] Else, the simultaneous mapping must not influence the overall cost and so its cost is always set to $0$.
\end{enumerate}
By using the edit cost functions defined in Section~\ref{sec:mainDefEditDist}, the cost of an allowed simultaneous node mapping is then defined by
\begin{equation}\label{eq-cikjl}
	\begin{aligned}
	c_e(i\,{\rightarrow}\,k,j\,{\rightarrow}\,l)=\,&\,c_{{es}}\left((i,j)\,{\rightarrow}\,(k,l)\right)\delta_{(i,j)\in E_1}\delta_{(k,l)\in E_2}\\
		&+c_{{ed}}\left(i,j\right)\delta_{(i,j)\in E_1}(1\,{-}\,\delta_{(k,l)\in E_2})\\
		&+c_{{ei}}\left(k,l\right)(1\,{-}\,\delta_{(i,j)\in E_1})\delta_{(k,l)\in E_2}
	\end{aligned}
\end{equation}
where $\delta_{e\,{\in}\,E}\,{=}\,1$ if $e\,{\in}\,E$ and $0$ else, $(i,j)\,{\rightarrow}\,\epsilon$ denotes edge removal and $\epsilon\,{\rightarrow}\,(k,l)$ denotes edge insertion. Since graphs do not have self-loops, we have $d_{ik,ik}\,{=}\,0$ for all $i\,{\in}\,V_1^\epsilon$ and $k\,{\in}\,V_2^\epsilon$. Remark also that the symmetry of $c_e(i\rightarrow k,j\rightarrow l)$ depends on the one of edit operations and the one of directed edges when the two graphs are directed.

Finally, the cost of a simultaneous node mapping is given by
\begin{equation}\label{eq-Dfinal}
	d_{ik,jl}=\left\{\begin{aligned}
		&\omega&&\text{if}~\,(i\,{\not\rightarrow}\,k)\vee(j\,{\nrightarrow}\,l)\\
		&c_e(i\,{\rightarrow}\,k,j\,{\rightarrow}\,l)&&\text{else}
	\end{aligned}\right.
\end{equation}

Let $\mathbf{x}\,{\in}\,\text{vec}[\mathcal{A}_{n,m}^\sim]\,{\subset}\,\{0,1\}^{(n+m)^2}$ be the vectorization of the $\epsilon$-assignment matrix associated to $\psi$. All costs can be represented by a $(n\,{+}\,m)^2 \times (n\,{+}\,m)^2$ matrix $\mathbf{D}\,{=}\,(d_{ik,jl})_{i,k,j,l}$ such that $d_{ik,jl}x_{ik}x_{jl}\,{=}\,d_{i\psi(i),j\psi(j)}$ if $x_{ik}\,{=}\,x_{jl}\,{=}\,1$, and $0$ else. So each row and each column of $\mathbf{D}$, and $\mathbf{x}$, have the same organization of pairwise indices, and then the total cost of the simultaneous node assignment can be written in quadratic form as:
\begin{equation*}
	d(\psi)=\sum_{i=1}^{n+m}\sum_{k=1}^{m+n}\sum_{j=1}^{n+m}\sum_{l=1}^{m+n} d_{ik,jl}x_{ik}x_{jl}=\mathbf{x}^T\mathbf{D}\mathbf{x},
\end{equation*}
The cost matrix $\mathbf{D}$ can be decomposed as follows into blocks:
\begin{equation}\label{eq:matrixD}
	\mathbf{D} = \left(\begin{array}{ccc|ccc}
	\mathbf{D}^{1,1}&\cdots&\mathbf{D}^{1,n}&\mathbf{D}^{1,\epsilon_1}&\cdots&\mathbf{D}^{1,\epsilon_m}\\
	\vdots & \ddots &\vdots&\vdots & \ddots &\vdots\\
	\mathbf{D}^{n,1} &\cdots & \mathbf{D}^{n,n}&\mathbf{D}^{n,\epsilon_1} &\cdots & \mathbf{D}^{n,\epsilon_m}\\\hline
	\mathbf{D}^{\epsilon_1,1}& \cdots & \mathbf{D}^{\epsilon_1,n}&\mathbf{D}^{\epsilon_1,\epsilon_1}&\cdots&\mathbf{D}^{\epsilon_1,\epsilon_m}\\
	\vdots & \ddots &\vdots&\vdots & \ddots &\vdots\\
	\mathbf{D}^{\epsilon_m,1} &\cdots & \mathbf{D}^{\epsilon_m,n}&\mathbf{D}^{\epsilon_m,\epsilon_1}&\cdots&\mathbf{D}^{\epsilon_m,\epsilon_m}
	\end{array}\right)
\end{equation}
where each block $\mathbf{D}^{i,j}\,{\in}\,[0,{+\infty})^{(m+n)\times(m+n)}$ defines the cost of assigning $i$ and $j$ of $\,V_1^\epsilon$ to respectively $k$ and $l$ for all $k,l\,{\in}\,V_2^\epsilon$, \textit{i.e.} $[\mathbf{D}^{i,j}]_{k,l}\,{=}\,d_{ik,jl}$. Remark that blocks $\mathbf{D}^{i,j}$ are organized in four main blocks corresponding to the nature of nodes $i$ and $j$ (dummy nodes or not). Each block $\mathbf{D}^{i,j}$ is itself decomposed into four blocks as follows:
\begin{equation}\label{eq:matrixDb}
	\mathbf{D}^{i,j} =\begin{blockarray}{@{\hspace{2pt}}c|c@{\hspace{2pt}}cc}
	j1\cdots jm&j\epsilon_1\cdots j\epsilon_n &&\\
	\begin{block}{(@{\hspace{2pt}}c|c)@{\hspace{2pt}}cc}\vspace{-0.15cm}
		&&&i1\\\vspace{-0.15cm}
		\mathbf{D}^{i,j}_{1,1}& \mathbf{D}^{i,j}_{1,2} &&\vdots\\
		&  &&im\\\cline{1-4}\vspace{-0.15cm}
		& &&i\epsilon_1\\\vspace{-0.15cm}
		\mathbf{D}^{i,j}_{2,1}&\mathbf{D}^{i,j}_{2,2}	&&\vdots\\\vspace{0.15cm}
		&  &&i\epsilon_n\\
	\end{block}
	\end{blockarray}
\end{equation}
where $\mathbf{D}^{i,j}_{1,1}\,{\in}\,[0,{+\infty})^{m\times m}$, $\mathbf{D}^{i,j}_{1,2}\,{\in}\,[0,{+\infty})^{m\times n}$, $\mathbf{D}^{i,j}_{2,1}\,{\in}\,[0,{+\infty})^{n\times m}$ and $\mathbf{D}^{i,j}_{2,2}\,{\in}\,[0,{+\infty})^{n\times n}$. The different values taken by the elements of $\mathbf{D}$, depending on the values of the indices, are reported in Table~\ref{tab:qap_formulation}.
\begin{table}[!t]
  \begin{tabular}{|r|r|l|l|}
	\hline
     case&block&nodes in $V_2^\epsilon$&\multicolumn{1}{|c|}{$d_{(i,j),(k,l)}$}\\
    \hline
    1&$\mathbf{D}^{i,j}_{1,1}$ &$k,l$&\begin{tabular}{l}
               $c_{es}((i,j)\,{\rightarrow}\,(k,l))\delta_{(i,j)\in E_1}\delta_{(k,l)\in E_2}\,+$\\
               $c_{ed}(i,j)\delta_{(i,j)\in E_1}(1-\delta_{(k,l)\in E_2})\,+$\\
               $c_{ei}(k,l)(1-\delta_{(i,j)\in E_1})\delta_{(k,l)\in E_2}$\end{tabular}\\
    \hline
    \multirow{2}{*}{2}&\multirow{2}{*}{$\mathbf{D}^{i,j}_{1,2}$}&$k,\epsilon_j$&\begin{tabular}{l}$c_{ed}(i,j)\delta_{(i,j)\in E_1}$\end{tabular}\\\cline{3-4}
    &&else&\begin{tabular}{l}$\omega$\end{tabular}\\\hline
    \multirow{2}{*}{3}&\multirow{2}{*}{$\mathbf{D}^{i,j}_{2,1}$}&$\epsilon_i,l$&\begin{tabular}{l}$c_{ed}(i,j)\delta_{(i,j)\in E_1}$\end{tabular}\\\cline{3-4}
    &&else&\begin{tabular}{l}$\omega$\end{tabular}\\\hline
    \multirow{2}{*}{4}&\multirow{2}{*}{$\mathbf{D}^{i,j}_{2,2}$}&$\epsilon_i,\epsilon_j$&\begin{tabular}{l}$c_{ed}(i,j)\delta_{(i,j)\in E_1}$\end{tabular}\\\cline{3-4}
    &&else&\begin{tabular}{l}$\omega$\end{tabular}\\
    \hline
    \multirow{2}{*}{5}&\multirow{2}{*}{$\mathbf{D}^{i,\epsilon_l}_{1,1}$}&$k,l$&\begin{tabular}{l}$c_{ei}(k,l)\delta_{(k,l)\in E_2}$\end{tabular}\\\cline{3-4}
    &&else&\begin{tabular}{l}$\omega$\end{tabular}\\
    \hline
    6&$\mathbf{D}^{i,\epsilon_l}_{1,2}$&$k,\epsilon$&\begin{tabular}{l}$0$\end{tabular}\\\hline
    \multirow{2}{*}{7}&\multirow{2}{*}{$\mathbf{D}^{i,\epsilon_l}_{2,1}$}&$\epsilon_i,l$&\begin{tabular}{l}$0$\end{tabular}\\\cline{3-4}
    &&else&\begin{tabular}{l}$\omega$\end{tabular}\\
    \hline
    \multirow{2}{*}{8}&\multirow{2}{*}{$\mathbf{D}^{i,\epsilon_l}_{2,2}$}&$\epsilon_i,\epsilon$&\begin{tabular}{l}$0$\end{tabular}\\\cline{3-4}
    &&else&\begin{tabular}{l}$\omega$\end{tabular}\\
    \hline
    \multirow{2}{*}{9}&\multirow{2}{*}{$\mathbf{D}_{1,1}^{\epsilon_k,j}$}&$k,l$&\begin{tabular}{l}$c_{ei}\left(k,l\right)\delta_{(k,l)\in E_2}$\end{tabular}\\\cline{3-4}
    &&else&\begin{tabular}{l}$\omega$\end{tabular}\\
    \hline
    \multirow{2}{*}{10}&\multirow{2}{*}{$\mathbf{D}_{1,2}^{\epsilon_k,j}$}&$k,\epsilon_j$&\begin{tabular}{l}$0$\end{tabular}\\\cline{3-4}
    &&else&\begin{tabular}{l}$\omega$\end{tabular}\\
    \hline
    11&$\mathbf{D}_{2,1}^{\epsilon_k,j}$&$\epsilon,l$&\begin{tabular}{l}$0$\end{tabular}\\
    \hline
    \multirow{2}{*}{12}&\multirow{2}{*}{$\mathbf{D}_{2,2}^{\epsilon_k,j}$}&$\epsilon,\epsilon_j$&\begin{tabular}{l}$0$\end{tabular}\\\cline{3-4}
    &&else&\begin{tabular}{l}$\omega$\end{tabular}\\
    \hline
    \multirow{2}{*}{13}&\multirow{2}{*}{$\mathbf{D}_{1,1}^{\epsilon_k,\epsilon_l}$}&$k,l$&\begin{tabular}{l}$c_{ei}\left(k,l\right)\delta_{(k,l)\in E_2}$\end{tabular}\\\cline{3-4}
    &&else&\begin{tabular}{l}$\omega$\end{tabular}\\
    \hline
    \multirow{2}{*}{14}&\multirow{2}{*}{$\mathbf{D}_{1,2}^{\epsilon_k,\epsilon_l}$}&$k,\epsilon$&\begin{tabular}{l}$0$\end{tabular}\\\cline{3-4}
    &&else&\begin{tabular}{l}$\omega$\end{tabular}\\
    \hline
    \multirow{2}{*}{15}&\multirow{2}{*}{$\mathbf{D}_{2,1}^{\epsilon_k,\epsilon_l}$}&$\epsilon,l$&\begin{tabular}{l}$0$\end{tabular}\\\cline{3-4}
    &&else&\begin{tabular}{l}$\omega$\end{tabular}\\
    \hline
    16&$\mathbf{D}_{2,2}^{\epsilon_k,\epsilon_l}$&$\epsilon,\epsilon$&\begin{tabular}{l}$0$\end{tabular}\\
    \hline
  \end{tabular}
\caption{Elements of matrix $D$ according to the configuration of its indices. We consider that $i,j\,{\in}\,V_1$, $k,l\,{\in}\,V_2$, $\epsilon_k,\epsilon_l\,{\in}\,\mathcal{E}_1$, and $\epsilon_i,\epsilon_j\,{\in}\,\mathcal{E}_2$. Epsilon values without indices mean any $\epsilon$-value.}\label{tab:qap_formulation}
\end{table}
\begin{proposition}\label{prop:Dsymm}
  If both $G_1$ and $G_2$ are undirected, then:
  \[
  \forall (i,k,j,l)\in V_1^\epsilon\times V_2^\epsilon\times V_1^\epsilon\times V_2^\epsilon,~~ d_{ik,jl}=d_{jl,ik}.
  \]
\end{proposition}
\begin{proof}
  If both $G_1$ and $G_2$ are undirected, then
  $\delta_{(i,j)\in E_1}=\delta_{(j,i)\in E_1}$,
  $\delta_{(k,l)\in E_2}=\delta_{(l,k)\in E_2}$ and :
  \[
  \left\{
    \begin{array}{lll}
      c_{es}((i,j)\rightarrow(k,l))&=& c_{es}((j,i)\rightarrow(l,k))\\
      c_{ed}(i,j)&=&      c_{ed}(j,i)\\
      c_{ei}(k,l)&=&      c_{ei}(l,k)\\
    \end{array}
    \right.
    \]
    Hence if none of $i,j,k$ or $l$ are equal to $\epsilon$, the first line of
    Table~\ref{tab:qap_formulation} remains unchanged. Moreover, if
    $\epsilon\in \{i,j,k,l\}$, then permuting indices $(i,k)$ and
    $(j,l)$ leads to the following permutations of the lines of
    Table~\ref{tab:qap_formulation} (after the appropriate renaming
    of variables):
    \[
    (2,3)(4)(5,9)(6,11)(7,10)(8,12)(13)(14,15)(16)
    \]
    One can check that in each case $d_{ik,jl}=d_{jl,ik}$.
\end{proof}
\begin{remarque}
  If the rows of matrix $\mathbf{D}$ correspond to $(i,k)$ and the
  columns to $(j,l)$, then under the hypothesis of
  Proposition~\ref{prop:Dsymm}, $\mathbf{D}$ is symmetric and we get
  $\mathbf{D}^T=\mathbf{D}$.
\end{remarque}

To fully represent edit operations we also need to consider the ones performed on nodes. This can be measured by the linear sum $\mathbf{c}^T\mathbf{x}$ defined in Section~\ref{sec-lsapeps}, where $\mathbf{c}=\text{vec}(\mathbf{C})\in[0,{+\infty})^{(n+m)^2}$ represents the cost of edit operations on nodes (Eq.~\ref{eq:cost-matrix}):
\begin{equation}\label{eq-costlin}
	\begin{aligned}
	&c^\text{sub}_{i,k}=c_{{vs}}(i\,{\rightarrow}\,k)\\
	&c^\text{rem}_{i,k}=\left\lbrace\begin{aligned}
		&c_{{vd}}(i)&&\text{if}~k\,{=}\,i\\
		&\omega&&\text{else}
	\end{aligned}\right.\\
	&c^\text{ins}_{i,k}=\left\{\begin{aligned}
		&c_{{vi}}(k)&&\text{if}~i\,{=}\,k\\
		&\omega&&\text{else}
	\end{aligned}\right.
	\end{aligned}
\end{equation}
\subsection{QAP for $\epsilon$-assignments, restricted edit paths and GED}
According to the following result, summing the quadratic and the linear costs defined above leads to the cost of a restricted edit path.
\begin{proposition}\label{prop:qap_ged}
  Let $\Delta=\mathbf{D}$ if both $G_1$ and $G_2$ are undirected and $\Delta=\mathbf{D}+\mathbf{D}^T$ else. Note that using Proposition~\ref{prop:Dsymm}, $\Delta$ is symmetric. Any non-infinite value of $\tfrac{1}{2}\mathbf{x}^T\Delta\mathbf{x}+\mathbf{c}^T\mathbf{x}$ corresponds
  to the cost of a minimal edit path. Conversely the cost of any
  minimal edit path may be written as $\tfrac{1}{2}\mathbf{x}^T\Delta\mathbf{x}+\mathbf{c}^T\mathbf{x}$ with the
  appropriate $\mathbf{x}$.
\end{proposition}
\begin{proof}
  If $\tfrac{1}{2}\mathbf{x}^T\Delta\mathbf{x}+\mathbf{c}^T\mathbf{x}$ is not infinite, $\mathbf{x}$ corresponds to a matrix of
  assignment (Definition~\ref{def:assMatrix}). Hence by
  propositions~\ref{prop:corresp_AmMappings}
  and~\ref{prop:edit_path_mapping}, there is a unique restricted edit
  path $P$ such that $P$ and $\mathbf{x}$ correspond to the same injective
  function $\varphi$ from a subset $\hat{V}_1$ of $V_1$ onto
  $V_2$. The edit path $P$ substitutes each node $v$, belonging to
  $\hat{V}_1$ onto $\varphi(v)$, removes all nodes belonging to
  $V_1\setminus\hat{V}_1$ and insert all nodes belonging to
  $V_2\setminus\varphi[\hat{V}_1]$ 
  (Proposition~\ref{prop:corresp_AmMappings}). 
  
	Consider the notations of the proof of Proposition~\ref{prop:corresp_AmMappings}, \textit{i.e.} the matrix $\mathbf{X}=(\mathbf{Q},\mathbf{R},\mathbf{S},\mathbf{T})$ defining $\mathbf{x}$. The three first blocks are defined as: 
  $q_{i,\varphi(i)}=1$ for any $i\in \hat{V}_1$, $r_{i,i}=1$
  for any $i\in V_1\setminus\hat{V}_1$ and $s_{j,j}=1$ for any
  $j\in V_2\setminus\varphi[\hat{V}_1]$. These blocks are filled with $0$
  outside these indices (Proposition~\ref{prop:edit_path_mapping}).

  Let us decompose the cost $\gamma(P)$ of
  $P$ into a cost $\gamma_v(P)$ related to operations on nodes and
  a cost $\gamma_e(P)$ related on operation on edges:
  \[
  \gamma(P)=\gamma_v(P)+\gamma_e(P)
  \]
  We have according to corollary~\ref{coro:costEditPath}:
  \[
  \gamma_v(P) =     \sum_{v\in V_1\setminus\hat{V}_1} c_{vd}(v)+
  \sum_{v\in \hat{V}_1} c_{vs}(v)+ \sum_{v\in V_2\setminus\hat{V}_2} c_{vi}(v)
  \]
  On the other hand, let us denote by $\mathbf{C}^\text{sub}$, $\mathbf{C}^\text{rem}$ and
  $\mathbf{C}^\text{ins}$ the blocks corresponding to $(\mathbf{Q},\mathbf{R},\mathbf{S})$ in the cost matrix $\mathbf{C}$ defining $\mathbf{c}$ (Section~\ref{sec-lsapeps}). The term $\mathbf{c}^T\mathbf{x}$ is equal to:
  \[
  \begin{aligned}
    \mathbf{c}^T\mathbf{x}&=\sum_{i=1}^n\sum_{j=1}^m c^\text{sub}_{i,j}q_{i,j}+
            \sum_{i=1}^n\sum_{j=1}^n c^\text{rem}_{i,j}r_{i,j}+
            \sum_{i=1}^m\sum_{j=1}^m c^\text{ins}_{i,j}s_{i,j}\\
    &=\sum_{i\in \hat{V}_1} c^\text{sub}_{i,\varphi(i)}+
        \sum_{i\in V_1\setminus\hat{V}_1} c^\text{rem}_{i,i}+
        \sum_{j\in V_2\setminus\varphi[\hat{V}_1]} c^\text{ins}_{j,j}\\
    &=\sum_{i\in \hat{V}_1} c_{vs}(i\rightarrow \varphi(i))+
        \sum_{i\in V_1\setminus\hat{V}_1} c_{vd}(i)+
        \sum_{j\in V_2\setminus\varphi[\hat{V}_1]} c_{vi}(j)
  \end{aligned}
  \]
  The first line takes into account that the block corresponding to
  $\mathbf{T}$ in matrix $\mathbf{C}$ is equal to $\mathbf{0}$. Thus, this block does not play any
  role in the computation of $\mathbf{c}^T\mathbf{x}$. The last line is due to the
  definition of the remaining blocks of $\mathbf{C}$
  (equation~\ref{eq-costlin}).

  We have thus $\gamma_v(P)=\mathbf{c}^T\mathbf{x}$. Let us show that
  $\gamma_e(P)=\tfrac{1}{2}\mathbf{x}^T\Delta\mathbf{x}$. We have:
  \[
  \begin{aligned}
  \left[\mathbf{x}^T\mathbf{D}\right]_{jl}&=\sum_{i,k}d_{ik,jl}x_{ik}=\sum_{i=1}^n\sum_{k=1}^m d_{ik,jl}q_{i,k}+
        \sum_{i=1}^n\sum_{k=1}^m d_{ik,jl}r_{i,k}+
        \sum_{k=1}^m\sum_{i=1}^n d_{ik,jl}s_{i,k}\\
    &=\sum_{i=1}^n\sum_{k=1}^m d_{ik,jl}q_{i,k}+
        \sum_{i=1}^n d_{i\epsilon,jl}r_{i,i}+
        \sum_{k=1}^m d_{\epsilon k,jl}s_{k,k}\\
    &=\sum_{i\in \hat{V}_1} d_{i\varphi(i),jl}+
        \sum_{i\in V_1\setminus\hat{V}_1} d_{i\epsilon,jl}+
        \sum_{k\in V_2\setminus\hat{V}_2} d_{\epsilon k,jl}
  \end{aligned}
\]
where $\hat{V_2}=\varphi[\hat{V}_1]$ (Proposition~\ref{prop:defSRIsets}). In the same way we have:
\[
  \mathbf{x}^T\mathbf{D}\mathbf{x}=\sum_{j,l}\left[\mathbf{x}^T\mathbf{D}\right]_{jl}x_{jl}=\sum_{j\in \hat{V}_1} \left[\mathbf{x}^T\mathbf{D}\right]_{j\varphi(j)}+
        \sum_{j\in V_1\setminus\hat{V}_1} \left[\mathbf{x}^T\mathbf{D}\right]_{j\epsilon}+
        \sum_{l\in V_2\setminus\hat{V}_2} \left[\mathbf{x}^T\mathbf{D}\right]_{\epsilon l}
\]
We have:
\[
  \begin{aligned}
  \sum_{j\in \hat{V}_1} \left[\mathbf{x}^T\mathbf{D}\right]_{j\varphi(j)}    
    &=\sum_{j\in \hat{V}_1}\left(\sum_{i\in \hat{V}_1} d_{i\varphi(i),j\varphi(j)}+
        \sum_{i\in V_1\setminus\hat{V}_1} d_{i\epsilon,j\varphi(j)}+
        \sum_{k\in V_2\setminus\hat{V}_2} d_{\epsilon k,j\varphi(j)}\right)\\
    &=\underset{(i,j)\in \hat{V}_1^2}{\sum} d_{i\varphi(i),j\varphi(j)}+
        \underset{(i,j)\in E_1\cap (V_1\setminus\hat{V}_1)\times \hat{V}_1}{\sum}\hspace*{-0.7cm}c_{ed}(i,j)~+\underset{(k,l)\in E_2\cap (V_2\setminus\hat{V}_2)\times \hat{V}_2}{\sum}\hspace*{-0.7cm}c_{ei}(k,l)\\
\\
    \sum_{j\in V_1\setminus\hat{V}_1} \left[\mathbf{x}^T\mathbf{D}\right]_{j\epsilon}&=
    \sum_{j\in V_1\setminus\hat{V}_1}\left(\sum_{i\in \hat{V}_1} d_{i\varphi(i),j\epsilon}+
        \sum_{i\in V_1\setminus\hat{V}_1} d_{i\epsilon,j\epsilon}+
        \sum_{k\in V_2\setminus\hat{V}_2} d_{\epsilon k,j\epsilon}\right)\\
    &=\underset{(i,j)\in E_1 \cap V_1\times (V_1\setminus\hat{V}_1)}{\sum}\hspace*{-0.7cm}c_{ed}(i,j)\\
\\
    \sum_{l\in V_2\setminus\hat{V}_2} \left[\mathbf{x}^T\mathbf{D}\right]_{\epsilon l}&=\sum_{l\in V_2\setminus\hat{V}_2} \left(\sum_{i\in \hat{V}_1} d_{i\varphi(i),\epsilon l}+
    \sum_{i\in V_1\setminus\hat{V}_1} d_{i\epsilon,\epsilon l}+
    \sum_{k\in V_2\setminus\hat{V}_2} d_{\epsilon k,\epsilon l}\right)\\
&=\underset{(k,l)\in E_2 \cap  V_2\times(V_2\setminus\hat{V}_2)}{\sum}\hspace*{-0.7cm}c_{ei}(k,l)\\
  \end{aligned}
\]
where all substitutions of $d_{ij,kl}$ by the corresponding removal
or insertion costs are deduced from
Table~\ref{tab:qap_formulation}.

Moreover, we have, using the notations of Proposition~\ref{prop:defSRIsets}:
\[
\begin{array}{lll}
  \{(i,j)\in \hat{V_1}^2~|~\delta_{(i,j)\in E_1}=\delta_{(\varphi(i),\varphi(j))\in E_2}=1\}&=&\hat{E_1}\\
    \{(i,j)\in \hat{V_1}^2~|~\delta_{(i,j)\in E_1}=1, \delta_{(\varphi(i),\varphi(j))\in E_2}=0\}&=&R_{12}\\
    \{(k,l)\in \hat{V_2}^2~|~\delta_{(\varphi^{-1}(k),\varphi^{-1}(l))\in E_1}=0, \delta_{(k,l)\in E_2}=1\}&=&I_{21}\\
\end{array}
\]
Hence:
\[
\sum_{j\in \hat{V}_1}\sum_{i\in \hat{V}_1} d_{i\varphi(i),j\varphi(j)} =
\sum_{(i,j)\in \hat{E}_1}c_{es}((i,j)\rightarrow (\varphi(i),\varphi(j)))+
\sum_{(i,j)\in R_{12}}c_{ed}(i,j)+
\sum_{(k,l)\in I_{21}}c_{ei}(k,l)
\]
Moreover we have:
\[
(V_1\setminus\hat{V_1})\times\hat{V_1}\cup V_1\times(V_1\setminus\hat{V_1})=(V_1\setminus\hat{V_1})\times V_1\cup V_1\times (V_1\setminus\hat{V_1})
\]
Indeed:
\[
\begin{array}{lll}
(V_1\setminus\hat{V_1})\times V_1\cup V_1\times (V_1\setminus\hat{V_1})  &=&
(V_1\setminus\hat{V_1})\times \hat{V_1}\cup
(V_1\setminus\hat{V_1})\times (V_1\setminus\hat{V_1})\cup
V_1\times (V_1\setminus\hat{V_1})\\
&=&(V_1\setminus\hat{V_1})\times \hat{V_1}\cup
V_1\times (V_1\setminus\hat{V_1})\\
\end{array}
\]
Hence by grouping appropriate terms we get:
\[
\begin{aligned}
  \mathbf{x}^T\mathbf{D}\mathbf{x}=&\sum_{(i,j)\in \hat{E}_1}c_{es}((i,j)\rightarrow (\varphi(i),\varphi(j)))\\
        &+\sum_{(i,j)\in E_1\cap \left((V_1\setminus\hat{V_1})\times V_1\cup V_1\times (V_1\setminus\hat{V_1})\right)\cup R_{12}}\hspace*{-2cm} c_{ed}(i,j)~~~~~~~~~+\sum_{(k,l)\in E_2\cap\left((V_2\setminus\hat{V_2})\times V_2\cup V_2\times (V_2\setminus\hat{V_2})\right)\cup I_{21}}\hspace*{-2cm}c_{ei}(k,l)\\
\end{aligned}
\]
Using Proposition~\ref{prop:defSRIsets} we finally get:
\[
  \mathbf{x}^T\mathbf{D}\mathbf{x}=\sum_{(i,j)\in \hat{E}_1}c_{es}((i,j)\rightarrow (\varphi(i),\varphi(j)))~+\sum_{(i,j)\in E_1\setminus\hat{E_1}}\hspace*{-0.3cm}c_{ed}(i,j)~+\sum_{(k,l)\in E_2\setminus\hat{E_2}}\hspace*{-0.3cm}c_{ei}(k,l)
\]
Using Corollary~\ref{coro:costEditPath} if both $G_1$ and $G_2$ are
directed we get $\gamma_e(P)=\mathbf{x}^T\mathbf{D}\mathbf{x}$. However in this case:
\[
\frac{1}{2}\mathbf{x}^T\Delta\mathbf{x}=\frac{1}{2}\mathbf{x}^T\left(\mathbf{D}+\mathbf{D}^T\right)\mathbf{x}=
\frac{1}{2}\left(\mathbf{x}^T\mathbf{D}\mathbf{x}+\mathbf{x}^T\mathbf{D}^T\mathbf{x}\right)=
\frac{1}{2}\left(\mathbf{x}^T\mathbf{D}\mathbf{x}+\mathbf{x}^T\mathbf{D}\mathbf{x}\right)=\mathbf{x}^T\mathbf{D}\mathbf{x}
\]
Hence $\gamma_e(P)=\frac{1}{2}\mathbf{x}^T\Delta\mathbf{x}$.

Using Remark~\ref{rem:doubleCounting} we get
$\gamma_e(P)=\frac{1}{2}\mathbf{x}^T\mathbf{D}\mathbf{x}=\frac{1}{2}\mathbf{x}^T\Delta\mathbf{x}$ if both $G_1$
and $G_2$ are undirected. 

Thus $\gamma(P)=\frac{1}{2}\mathbf{x}^T\Delta\mathbf{x}+\mathbf{c}^T\mathbf{x}$ with
$\Delta=\mathbf{D}$ if both $G_1$ and $G_2$ are undirected and $\Delta=\mathbf{D}+\mathbf{D}^T$
else.
\end{proof}

~\\\indent
Hence, the determination of a restricted edit path with a minimal cost is equivalent to searching for an optimal $\epsilon$-assignment 
\begin{equation}\label{eq:qap_formulation}
  \hat{\mathbf{x}} \in \argmin\left\{\,\frac{1}{2}\mathbf{x}^T\Delta\mathbf{x} + \mathbf{c}^T\mathbf{x}~|~\mathbf{x}\,{\in}\,\text{vec}[\mathcal{A}^\sim_{n,m}]\,\right\}
\end{equation}
In other terms, for the class of graphs under consideration, \textit{i.e.} simple graphs, we have
\begin{equation}
	\text{GED}(G_1,G_2)=\min\left\{\,\frac{1}{2}\mathbf{x}^T\Delta\mathbf{x} + \mathbf{c}^T\mathbf{x}~|~\mathbf{x}\,{\in}\,\text{vec}[\mathcal{A}^\sim_{n,m}]\,\right\}
\end{equation}
This is a QAP, see \cite{Lowler1963,bur09} for more details on QAPs. In
particular, QAPs are NP-hard and exact algorithms can solve QAPs of small size only. So, 
many heuristics able to find suboptimal solutions in short computing time have been explored. 
\begin{remarque}
Note that the functional involved in the QAP defined by Eq.~\ref{eq:qap_formulation} can be rewritten as a general quadratic term:
\begin{equation}
	\frac{1}{2}\mathbf{x}^T\Delta\mathbf{x} + \mathbf{c}^T\mathbf{x}=\frac{1}{2}\mathbf{x}^T\Delta\mathbf{x} + \mathbf{x}^T\text{diag}(\mathbf{c})\mathbf{x}=\mathbf{x}^T\left(\frac{1}{2}\Delta+\text{diag}(\mathbf{c})\right)\mathbf{x}
\end{equation}
where $\text{diag}(\mathbf{c})$ is the diagonal matrix with $\mathbf{c}$ as diagonal. So the GED can be equivalently defined by 
\begin{equation}
	\text{GED}(G_1,G_2)=\min\left\{\,\mathbf{x}^T\overline{\Delta}\mathbf{x}~|~\mathbf{x}\,{\in}\,\text{vec}[\mathcal{A}^\sim_{n,m}]\,\right\}
\end{equation}
where $\overline{\Delta}\,{=}\,\tfrac{1}{2}\Delta\,{+}\,\text{diag}(\mathbf{c})$ represents the cost of both node and edge edit operations. As graphs are simple, they have no self-loops and then the diagonal elements of $\Delta$ are all equal to $0$. So the diagonal of $\overline{\Delta}$ is always equal to $\mathbf{c}$.
\end{remarque}

\section{Solving QAPs with the Integer Projected Fixed Point Algorithm}
\label{sec:qap_solvers}
We propose to compute an approximate GED by finding a solution of the QAP defined by Eq.~\ref{eq:qap_formulation}, and rewritten here as the following binary quadratic programming problem:
\begin{equation}\label{eq-qapsym}
	\argmin\left\{S(\mathbf{x})\overset{\text{def.}}{=}\frac{1}{2}\mathbf{x}^T\Delta\mathbf{x}+\mathbf{c}^T\mathbf{x}~|~\mathbf{A}\mathbf{x}\,{=}\,\mathbf{1}_n,~\mathbf{x}\,{\in}\,\{0,1\}^{n}\right\}
\end{equation}
where $\mathbf{A}\mathbf{x}\,{=}\,\mathbf{1}_n$, with $\mathbf{x}\,{\in}\,\{0,1\}^n$ and $\mathbf{A}\,{\in}\,\{0,1\}^{n\times n}$, is the matrix version of the bijectivity constraints given by Eq.~\ref{eq-cts-mtx-perm}, see \cite{bur09,sier01} for more details. Also, we suppose that $\mathbf{c}\,{\in}\,[0,{+\infty})^n$, and $\Delta\,{\in}\,[0,{+\infty})^{n\times n}$ is assumed to be symmetric. Note that Eq.~\ref{eq:qap_formulation} is equivalent to Eq.~\ref{eq-qapsym}, with additional constraints on $\mathbf{x}$ (Eq.~\ref{eq-cst-eps}) imposed by $\omega$ values in the expression of $\Delta$ (Eq.~\ref{eq-Dfinal} and Proposition~\ref{prop:qap_ged}) and $\mathbf{c}$ (Eq.~\ref{eq-costlin}).

QAPs are generally NP-hard, which depends on the structure of the cost matrix $\overline{\Delta}\,{=}\,\tfrac{1}{2}\mathbf{\Delta}\,{+}\,\text{diag}(\mathbf{c})$ (see previous section), and so most algorithms find approximate local or global optimal solutions by relaxing the bijectivity constraints on the solution, which leads to find a continuous solution instead of a discrete one:
\begin{equation}\label{eq-relaxed}
\argmin\left\{S(\mathbf{x})~|~\mathbf{A}\mathbf{x}\,{=}\,1,~\mathbf{x}\,{\in}\,[0,{+\infty})^n\right\}.
\end{equation}
While this relaxed problem is also NP-hard, several polynomial-time
algorithms have been designed to converge close to a local or global
solution in a short computing time. The ones based on linearization of
the cost function $S$ are known to be particularly efficient. They
transform the relaxed problem into a sequence of convex problems, such
that a given initial solution is improved iteratively by decreasing
the cost function up to a fixed point. Then, the final continuous
solution is binarized and used as a solution of the QAP. But as shown
experimentally in \cite{Leordeanu2009} in the context of graph
matching, the continuous optimum is not necessarily close to the
global discrete optimum.  To try to overcome this
problem, it seems to be more efficient to try to find a discrete
solution as close as possible to a continuous one, at each iteration,
as done by Sof-Assign \cite{Gold1994,Gold1996,Gold1996a} or Integer
Projected Fixed Point (IPFP) \cite{Leordeanu2009}.

We present here an adaptation of the IPFP algorithm, originally proposed for maximization of a quadratic term \cite{Leordeanu2009}, to minimization. We also improve the computational complexity of several steps of the algorithm.

Given an initial continuous (or discrete) candidate solution $\mathbf{x}_0$,  The idea of \cite{Leordeanu2009} is to iteratively improve (here reduce) the corresponding quadratic cost in two steps at each iteration:
\begin{enumerate}
\item Compute a discrete linear approximation $\mathbf{b}_{k+1}$ of the quadratic cost $S$ around the current solution $\mathbf{x}_k$ by solving a LSAP.
\item Compute the next candidate solution $\mathbf{x}_{k+1}$ by
  solving the relaxed problem, reduced to compute the extremum of $S$
  between $\mathbf{x}_{k}$ and $\mathbf{b}_{k+1}$ included.
\end{enumerate}
The iteration of these steps converges to an optimum of the relaxed
problem, which is either continuous or discrete but generally not the global one. This last point depends on the initialization. The whole process is
detailed in Algorithm~\ref{alg:ipfp_opt}.

\begin{algorithm}[!t]
\begin{algorithmic}[1]
\State $k\leftarrow 0$,~~$L\leftarrow\mathbf{c}^T\mathbf{x}_0$,~~$S_k\leftarrow\tfrac{1}{2}\mathbf{x}_0^T\Delta\mathbf{x}_0 + L$
\Repeat
\State {\color{gray}//~\,Projection of the cost by solving a LSAP}
\State $\mathbf{b}_{k+1}\leftarrow \argmin\{(\mathbf{x}_k^T\Delta+\mathbf{c}^T)\mathbf{b}~|~\mathbf{b}\,{\in}\,\mathcal{A}_{n,m}^\sim \}$
\State {\color{gray}//~\,Minimize the quadratic cost along the direction $\mathbf{b}_{k+1}$}
\State $L^\prime\leftarrow \mathbf{c}^T\mathbf{b}_{k+1}$
\State $S_{\mathbf{b}_{k+1}} \leftarrow \tfrac{1}{2}\mathbf{b}_{k+1}^T \Delta \mathbf{b}_{k+1} + L^\prime$
\State $\alpha\leftarrow R(\mathbf{b}_{k+1})-2S_k+L$
\State $\beta\leftarrow S_{\mathbf{b}_{k+1}}+S_k-R(\mathbf{b}_{k+1})-L$
\State $t_0\leftarrow -\alpha/(2\beta)$
\If{$(\beta \leq 0)\vee(t_0\geq 1)$} 
\State $\mathbf{x}_{k+1} \leftarrow \mathbf{b}_{k+1}$,~~$S_{k+1} \leftarrow S_{\mathbf{b}_{k+1}}$,~~$L\leftarrow L^\prime$
\Else
\State $\mathbf{x}_{k+1} \leftarrow \mathbf{x}_k + t_0(\mathbf{b}_{k+1}-\mathbf{x}_k)$
\State $S_{k+1} \leftarrow S_k -\alpha^2/(4\beta)$
\State $L\leftarrow \mathbf{c}^T\mathbf{x}_{k+1}$
\EndIf
\State $k \leftarrow k + 1$
\Until{$(\mathbf{x}_{k+1} =\mathbf{x}_k)\vee(k \geq k_\text{max})$}
\State \textbf{if}~{$x_{k+1}=b_{k+1}$}~\textbf{then~return}~$(\mathbf{x}_{k+1},S_{k+1})$
\State $\mathbf{x}_{k+1}\leftarrow \argmin\{\mathbf{x}_{k+1}^T\mathbf{b}~|~\mathbf{b}\,{\in}\,\mathcal{A}_{n,m}^\sim \}$
\State \textbf{return}~$(\mathbf{x}_{k+1},S_{k+1})$
\end{algorithmic}
\label{alg:ipfp_opt}
\caption{~~IPFPmin$(\mathbf{x}_0,\mathbf{c},\Delta,k_\text{max})$}
\end{algorithm}
At each iteration, in the first step, the cost $S$ is linearly approximated. The differential of $S$ in $\mathbf{x}_k$ in the direction $\mathbf{h}$ is given by:
\[
DS(\mathbf{x}_k)\cdot\mathbf{h}=\mathbf{x}_k^T\Delta\mathbf{h}+\mathbf{c}^T\mathbf{h}~~~~\text{ (since }\Delta\text{ is symmetric).}
\]
Hence the first-order Taylor expansion of $S$ around the current solution $\mathbf{x}_k$ is given by:
\begin{equation}
\begin{aligned}
	S(\mathbf{b})&\approx S(\mathbf{x}_k)+\left(\mathbf{x}_k^T\Delta+\mathbf{c}^T\right)(\mathbf{b}\,{-}\,\mathbf{x}_k)\\
	&\approx S(\mathbf{x}_k)+R(\mathbf{b})-R(\mathbf{x}_k)
	\end{aligned}
\end{equation}
where $R(\mathbf{x})=(\mathbf{x}_k^T\Delta\,{+}\,\mathbf{c}^T)\mathbf{x}$ and $\mathbf{b}\,{\geq}\,\mathbf{0}$. Keeping $\mathbf{x}_k$ fixed, $S(\mathbf{x}_k)$ and $R(\mathbf{x}_k)$ are constant, and so the minimization of $S(\mathbf{b})$ is approximatively equivalent to the minimization of $R(\mathbf{b})$:
\begin{equation}
	\mathbf{b}_{k+1}\,{\in}\,\argmin\left\{\left(\mathbf{x}^T_k\Delta\,{+}\,\mathbf{c}^T\right)\mathbf{b}~|~\mathbf{A}\mathbf{b}\,{=}\,\mathbf{1},~\mathbf{b}\,{\geq}\,\mathbf{0}_n\right\}.
\end{equation}
This is a linear programming problem with totally unimodular constraint matrix $\mathbf{A}$ and the right-hand side vector of the linear system $\mathbf{A}\mathbf{b}\,{=}\,\mathbf{1}$ is integer valued. So, by standard tools in linear programming, there is an integer optimal solution, here binary and equal to the solution of the LSAP \cite{sier01,bur09}
\begin{equation}
	\mathbf{b}_{k+1}\,{\in}\,\argmin\left\{\left(\mathbf{x}^T_k\Delta\,{+}\,\mathbf{c}^T\right)\mathbf{b}~|~\mathbf{A}\mathbf{b}\,{=}\,\mathbf{1},~\mathbf{b}\,{\in}\,\{0,1\}^{n\times n}\right\}.
\end{equation}
In our experiments, this problem is solved by the $O(n^3)$ version of the Hungarian algorithm \cite{law76,bur09}, modified such that forbidden assignments represented by $\omega$ values in $\mathbf{\Delta}$ and $\mathbf{c}$ are not treated. The resulting assignment $\mathbf{b}_{k+1}$ determines a direction of largest possible decrease of $S$ in the first-order approximation. Let us additionally note that the first order approximation of $S(\mathbf{b})$
is lower than $S(\mathbf{x}_k)$ since $R(\mathbf{b}_{k+1})\leq R(\mathbf{x}_k)$. However we cannot yet conclude since this is only an approximation.

The second step of each iteration of Algorithm~\ref{alg:ipfp_opt} consists in minimizing the quadratic function $S$ in the continuous domain along the direction given by $\mathbf{b}_{k+1}$. This can be done analytically. Let $\mathbf{x}_t\,{=}\,\mathbf{x}_k\,{+}\,t(\mathbf{b}_{k+1}\,{-}\,\mathbf{x}_k)$, with $t\,{\in}\,[0,1]$, be a parameterization of the segment between $\mathbf{x}_{k}$ and $\mathbf{b}_{k+1}$. The evolution
of $S$ on this segment is provided by:
\begin{equation*}
\begin{aligned}
S(\mathbf{x}_t)=&~S(\mathbf{x}_k\,{+}\,t(\mathbf{b}_{k+1}\,{-}\,\mathbf{x}_k))\\
=&~\tfrac{1}{2}[\mathbf{x}_k+t(\mathbf{b}_{k+1}-\mathbf{x}_k)]^T\Delta{}[\mathbf{x}_k+t(\mathbf{b}_{k+1}-\mathbf{x}_k)]+\mathbf{c}^T[\mathbf{x}_k+t(\mathbf{b}_{k+1}-\mathbf{x}_k)]\\        
        =&~\tfrac{1}{2}\mathbf{x}_k^T\Delta{}\mathbf{x}_k+\mathbf{c}^T\mathbf{x}_k+t\mathbf{x}_k^T\Delta{}(\mathbf{b}_{k+1}-\mathbf{x}_k)+
        \tfrac{1}{2}t^2(\mathbf{b}_{k+1}-\mathbf{x}_k)^T\Delta{}(\mathbf{b}_{k+1}-\mathbf{x}_k)\\
        &~+t\mathbf{c}^T(\mathbf{b}_{k+1}-\mathbf{x}_k)\\
        =&~S(\mathbf{x}_k)+t[\mathbf{x}_k^T\Delta{}(\mathbf{b}_{k+1}-\mathbf{x}_k)+\mathbf{c}^T(\mathbf{b}_{k+1}-\mathbf{x}_k)]+
        \tfrac{1}{2}t^2(\mathbf{b}_{k+1}-\mathbf{x}_k)^T\Delta(\mathbf{b}_{k+1}-\mathbf{x}_k)\\
        =&~S(\mathbf{x}_k)+t R(\mathbf{b}_{k+1}-\mathbf{x}_k)+\tfrac{1}{2}t^2(\mathbf{b}_{k+1}-\mathbf{x}_k)^T\Delta(\mathbf{b}_{k+1}-\mathbf{x}_k)\\
=&~S(\mathbf{x}_k)+\alpha t+\beta t^2
\end{aligned}
\end{equation*}
where
\begin{equation*}
  \begin{aligned}
    \alpha=&~R(\mathbf{b}_{k+1}-\mathbf{x}_k)=R(\mathbf{b}_{k+1})-R(\mathbf{x}_k)\leq 0\\
  =&~R(\mathbf{b}_{k+1})-\mathbf{x}_k^T\Delta{}\mathbf{x}_k-\mathbf{c}^T\mathbf{x}_k=R(\mathbf{b}_{k+1})-2(\tfrac{1}{2}\mathbf{x}_k^T\Delta{}\mathbf{x}_k+\mathbf{c}^T\mathbf{x}_k)+\mathbf{c}^T\mathbf{x}_k\\
    =& ~R(\mathbf{b}_{k+1})-2S(\mathbf{x}_k)+\mathbf{c}^T\mathbf{x}_k\\
    \beta=&~\tfrac{1}{2}(\mathbf{b}_{k+1}-\mathbf{x}_k)^T\Delta(\mathbf{b}_{k+1}-\mathbf{x}_k)\\
    =&~\tfrac{1}{2}\mathbf{b}_{k+1}^T\Delta{}\mathbf{b}_{k+1}-\mathbf{x}_k^T\Delta{}\mathbf{b}_{k+1}+\tfrac{1}{2}\mathbf{x}_k^T\Delta{}\mathbf{x}_k\\
     =&~\tfrac{1}{2}\mathbf{b}_{k+1}^T\Delta{}\mathbf{b}_{k+1}+\mathbf{c}^T\mathbf{b}_{k+1}-R(\mathbf{b}_{k+1})+\tfrac{1}{2}\mathbf{x}_k^T\Delta{}\mathbf{x}_k\\
    =&~S(\mathbf{b}_{k+1})+S(\mathbf{x}_k)-R(\mathbf{b}_{k+1})-\mathbf{c}^T\mathbf{x}_k
  \end{aligned}
\end{equation*}
The main advantage of the above expression for the calculation of
$\alpha$, and not used in the original algorithm \cite{Leordeanu2009},
is that $R(\mathbf{b}_{k+1})$ is already computed by the LSAP
algorithm that computes $\mathbf{b}_{k+1}$. Moreover $S(\mathbf{x}_k)$
and $\mathbf{c}^T\mathbf{x}_k$ may be stored from the previous
iteration of the algorithm. Note that since
$\alpha=R(\mathbf{b}_{k+1})-R(\mathbf{x}_k)$ we have $\alpha\leq
0$. For $\beta$, the improvement is a bit more tedious. We have indeed
to compute $S(\mathbf{b}_{k+1})$. Hence the gain is not obvious
compared to the direct computation of
$(\mathbf{b}_{k+1}-\mathbf{x}_k)^T\Delta(\mathbf{b}_{k+1}-\mathbf{x}_k)$
as done in the original algorithm \cite{Leordeanu2009}. Note however
that $S(\mathbf{b}_{k+1})$ will also be used in a following step, so
computations are factorized.
The problem is thus transformed into finding the optimal value
\begin{equation}
	t_0=\argmin\left\{S(\mathbf{x}_t)=S(\mathbf{x}_k)+\alpha t+\beta t^2~|~t\,{\in}\,[0,1]\right\}.
\end{equation}
The derivative of $S(\mathbf{x}_t)$ cancels at $t_0=-\alpha/(2\beta)$. Then as shown in Fig.~\ref{fig:cases}, we have:
\begin{itemize}
\item If $\beta > 0$ 
  \begin{itemize}
  \item If $t_0\leq 1$, $S(\mathbf{x}_{t_0})$ is the minimum of $S(\mathbf{x}_t)$, in
    particular it is lower than $S(\mathbf{x}_k)$ and $S(\mathbf{b}_{k+1})$. Moreover:
    \[
    S(\mathbf{x}_{t_0})=S(\mathbf{x}_k)-\frac{\alpha^2}{2\beta}+\frac{\alpha^2}{4\beta}=S(\mathbf{x}_k)-\frac{\alpha^2}{4\beta}
    \]
  \item If $t_0\geq 1$, then $S'(\mathbf{x}_t)<0$ $\forall t\in [0,1]$, and the minimal value of $S(\mathbf{x}_t)$ is $S(\mathbf{x}_1)=S(\mathbf{b}_{k+1})$.
  \end{itemize}
\item   If $\beta\leq 0$, since $\alpha\leq 0$, $S(\mathbf{x}_t)$ decreases between
  $t=0$ and $t=1$. Its minimal value is thus $S(\mathbf{x}_1)=S(\mathbf{b}_{k+1})$.
\end{itemize}
So, if either $\beta\,{<}\,0$ or $\beta\,{>}\,0$, but $t_0\,{\geq}\,1$, the minima of $S(\mathbf{x}_t)$ within the range $t\,{\in}\,[0,1]$ is
obtained for $t_0\,{=}\,1$, \textit{i.e.} $\mathbf{x}_{t_0}\,{=}\,\mathbf{b}_{k+1}$ (lines 11-12). Note that in this case the new solution is discrete. In the remaining
case (lines 13-16), $S(\mathbf{x}_t)$ passes by a minimal value lower than $S(\mathbf{x}_k)$ and $S(\mathbf{b}_{k+1})$. In both cases $\mathbf{x}_{t_0}$ is taken as the solution $\mathbf{x}_{k+1}$ for the next iteration. Hence, as in the original algorithm \cite{Leordeanu2009}, $S(\mathbf{x}_k)$ decreases at each iteration, and since $\Delta$ and $\mathbf{c}$ are positive, $S$ is bounded bellow $0$ and the algorithm converges.
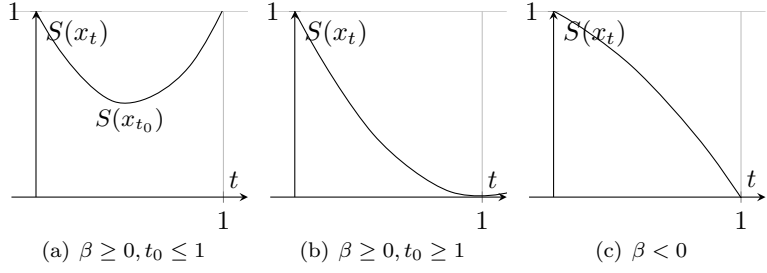
\begin{figure}[!t]
  \centering
  \subfigure[$\beta\geq 0, t_0\leq 1$]{\begin{tikzpicture}[baseline]
    \begin{axis}[ axis y line=center, axis x line=middle, axis equal,
      grid=both, xmax=1,xmin=0, ymin=0,ymax=1,
      xlabel=$t$,ylabel=$S(x_t)$, xtick={0,...,1}, ytick={0,...,1},
      width=.34\textwidth, anchor=center, ]
      \addplot[smooth] {2*x^2-2*x+1} ; 
      \node[above] at (axis cs:0.5,0.3){\small $S(x_{t_0})$};
    \end{axis}
  \end{tikzpicture}}
  \subfigure[$\beta\geq 0, t_0\geq 1$]{\begin{tikzpicture}[baseline]
    \begin{axis}[ axis y line=center, axis x line=middle, axis equal,
      grid=both, xmax=1,xmin=0, ymin=0,ymax=1,
      xlabel=$t$,ylabel=$S(x_t)$, xtick={0,...,1}, ytick={0,...,1},
      width=.34\textwidth, anchor=center, ]
      \addplot[smooth]{x^2-2*x+1} ; 
    \end{axis}
  \end{tikzpicture}}
  \subfigure[$\beta<0$]{\begin{tikzpicture}[baseline]
    \begin{axis}[ axis y line=center, axis x line=middle, axis equal,
      grid=both, xmax=1,xmin=0, ymin=0,ymax=1,
      xlabel=$t$,ylabel=$S(x_t)$, xtick={0,...,10}, ytick={0,...,10},
      width=.34\textwidth, anchor=center, ]
      \addplot[smooth] {-0.5*x^2-0.5*x+1} ;
    \end{axis}
  \end{tikzpicture}}
  \caption{Illustration of the 3 cases relating $\beta$ and $t_0$.}
\label{fig:cases}
\end{figure}

The whole process is iterated until a fixed point is reached, in which
case $\mathbf{x}_{k+1}$ is ensured to be a minimum of the relaxed
problem defined by Eq.~\ref{eq-relaxed}. When a minimum of the
original problem defined by Eq.~\ref{eq-qapsym} is requested as for
the GED, the discrete vector closest to the minimum of the relaxed
problem is selected. Note that the method~\cite{Justice2006} does not
guarantee the solution to be binary. The minimum of the relaxed
problem is not guaranteed to be global. This depends on the initial
vector $x_0$, which influences both the value of the resulting cost
and the number of iterations required to reach the convergence. For
the approximation of the GED, we have tested several initializations
based on the LSAP, as described in the following section. We have
observed that the exact GED is often obtained, meaning that the
optimal solution of the original problem can be reached by the
algorithm.


\section{Experiments}
\label{sec:experiments}
In this section, we present some experimentation results to show the
effectiveness of our quadratic assignment approach with respect
to the ones based on LSAPs. To this purpose, we compare our
approach to three methods based on the LSAP and one method to compute the
exact graph edit distance based on $A^*$ algorithm.
The latter method is used as a baseline to measure the error induced
by approximations. However, $A^*$ is restricted to very simple graphs
and it wasn't possible to compute exact graph edit distances for
two over the four used datasets. The three LSAP methods, \textit{i.e.} bipartite GEDs, are the ones proposed in \cite{Riesen2009,Gauzere2014a,Carletti2015}. They differ on the definition of the cost between elements, as already discussed. 

\begin{table}[!t]
\caption{Characteristics of the four GREYC's chemistry datasets.}\label{tab:datasets}
\centering
\begin{small}
\begin{tabular}{|l|c|c|c|c|c|}
\hline
\emph{Dataset} & \emph{Number of graphs} & \emph{Avg Size} & 
\emph{Avg Degree} \\
\hline \hline
$Alkane$		& 150 & 8.9 & 1.8 \\
\hline
$Acyclic$	& 183 & 8.2 & 1.8 \\
\hline
$MAO$		& 68 & 18.4 & 2.1 \\
\hline
$PAH$		& 94 & 20.7 & 2.4 \\
\hline
\end{tabular}
\end{small}
\end{table}
The experiments have been performed on four chemoinformatics datasets\footnote{Available at \url{https://brunl01.users.greyc.fr/CHEMISTRY/}}
(see Table~\ref{tab:datasets} for their characteristics) plus one
synthetic dataset to test larger graphs.  The variety of these
datasets allows to see the behavior of different approaches on four kind
of graphs: acyclic labeled (Acyclic), acyclic not labeled (Alkane),
cyclic labeled (MAO), cyclic not labeled (PAH). Finally, the synthetic
dataset is composed of graphs having same characteristics as the ones
in MAO dataset but extended up to 100 nodes (details bellow).

Table~\ref{tab:results} shows the results of our experiments on the
four GREYC's datasets. Note that, in order to avoid some bias, these
results have been computed using random permutations of the adjacency
matrices before computing graph edit distances. Moreover, all
experiments have been computed on the same machine, hence providing
comparable computational times. In order to show the relevancy of our
proposal, we compute the average edit distance ($d$), the average
approximation error ($e$) with respect to the exact graph edit
distance and the average computational time ($t$) required to get the
graph edit distance for a pair of graphs. For these three measures,
lower values correspond to better results. Indeed, since approximation
approaches overestimate graph edit distance, a lower
average edit distance can be considered as better than an higher
one. Due to the computational complexity required by A* algorithm, the
exact graph edit distance has not been computed on PAH and MAO
datasets which are composed of larger graphs than the ones in Acyclic
and Alkane datasets. 

\begin{table*}
\caption{Accuracy and complexity scores. $d$ is the average edit distance, $e$ the average error and
  $t$ the average computational time.}\label{tab:results}
\centering
\begin{small}
\begin{tabular}{lllllllllll}
  \toprule
  \multirow{2}{*}{Algorithm} & \multicolumn{3}{c}{Alkane} & \multicolumn{3}{c}{Acyclic}
  & \multicolumn{2}{c}{MAO} & \multicolumn{2}{c}{PAH}\\
  \cline{2-11} & $d$ & $e$ & $t$& $d$ & $e$ & $t$& $d$ & $t$& $d$ & $t$\\
  \midrule
  $A^*$&15  &&1.29  &17  &  &6.02 & & & &\\
  \cite{Riesen2009}&35  &18&$\simeq 10^{-3}$  &35  &18 &$\simeq 10^{-3}$ &105 &$\simeq 10^{-3}$ &138 &$\simeq 10^{-3}$\\
  \cite{Gauzere2014a}&33  &18&$\simeq 10^{-3}$  &31  &14 &$\simeq 10^{-2}$ &49 &$\simeq 10^{-2}$ &120 &$\simeq 10^{-2}$\\
  \cite{Carletti2015} &26  &11&2.27  &28  &9 &0.73 &44 &6.16  &129 &2.01\\
  \hline
  $IPFP_{\text{Random init}}$& 22.6 & 7.1& 0.007 &  23.4  & 6.1 &0.006& 65.2 &  0.031 &63 &0.04 \\
  $IPFP_{\text{init~\cite{Riesen2009} }}$& 22.4 & 7.0&  0.007 &  22.6  & 5.3
                     &0.006 & 59 &0.031 & 62.2 & 0.04 \\
  $IPFP_{\text{Init~\cite{Gauzere2014a}}}$& 20.5  & 5 & 0.006 & 20.7  & 3.4 &  0.005 & 33.6 & 0.016
                                      & 52.5& 0.037\\
   \bottomrule
\end{tabular}
\end{small}
\end{table*}
IPFP approach allows to drastically improve the accuracy of the
approximation with respect to LSAP approaches while keeping a
reasonable computational time. This observation can be made on the
four datasets, hence showing the consistency of this QAP approach. 
 
Results shown in Table~\ref{tab:results} also highlight the importance
of the initialization step. Since the functional of QAP formulation is
not convex, different initializations can lead to different local
minima. Obviously, initializations close to the global minimum have an
higher probability to reach it than initializations far from it. This
behavior is observed in the results where better approximations are
obtained by using the approach giving the best approximation
considering LSAP framework. Moreover, less iterations are required to
reach convergence since the algorithm is initialized close to the
minima. This phenomenon explains the low differences of computational
time between the different approaches. Note that we didn't test the
method presented in~\cite{Carletti2015} due to its high
computational time. In conclusion, these results show that the QAP
approach is a relevant approach to approximate graph edit distance and
outperforms methods based on LSAP formulation while keeping an
interesting computational time with respect to the one required to
compute an exact graph edit distance.

\begin{figure}[!t]
  \centering
  \includegraphics[width=0.9\linewidth]{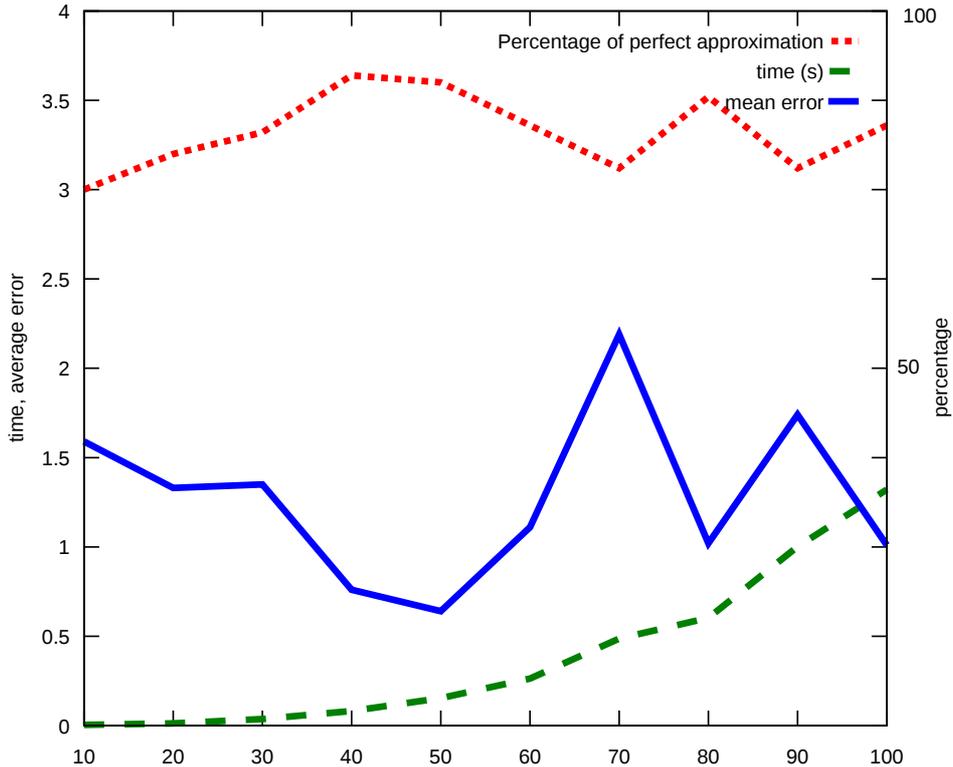}
  \caption{Analysis of complexity on a synthetic dataset.}\label{fig:time-analysis}
\end{figure}
Figure~\ref{fig:time-analysis} shows how our IPFP approach scales with
the size of graphs. This results have been computed on a synthetic
dataset having same node's and edge's labels distribution and same
ratio between the number of edges and the number of nodes as MAO
dataset but generalized to different graph sizes. For a given number
of nodes, a synthetic dataset is composed of 100 pairs of source and
target graphs. Each target graph has been generated by removing one
node and substituting another one from the associated source
graph. The overall edit distance between source and target graphs is
then defined by the cost associated to this two node operations
together with the induced edit operations on edges. The graph edit
distance between each pair of graphs is around 10. Given this
protocol, we generated a synthetic dataset composed of 100 couples of
graphs for different graph's sizes, hence obtaining 10 datasets from 10 to
100 nodes.

In Figure~\ref{fig:time-analysis}, the dashed green line corresponds
to the computational time required to compute an approximation of the
graph edit distance, the plain blue line to the average approximation
error and the dotted red line corresponds to the percentage of pairs
for which the exact graph edit distance is computed using IPFP. The
$x$ axis corresponds to the size of the graphs. The $y$ axis on the
left of Figure~\ref{fig:time-analysis} represents simultaneously the
mean execution times and the average error using a same
scaling. Hence, for example, $0.5$ should be read as $0.5$ seconds on
the dashed green line and as an average error of 0.5 on the plain blue
curve. The $y$ axis on the right corresponds to the percentage of
exact graph edit distances computed by our algorithm and should be
used for the analysis of the dotted red curve. As we can see, the
accuracy of the approximation using IPFP is stable for all tested
sizes. The average error for each dataset remains about 5 to 10 $\%$ of the exact graph edit
distance which corresponds to a good approximation. Moreover, the
percentage of perfect approximation shows  that we are able to
compute the exact graph edit distance for 75\% to 91\% pairs of
graphs. From a computational point of view, the curve seems to
describe a polynomial function with respect to the size of
graphs. Considering a bounded number of iterations, this observation
is conform with the cubic complexity associated to the algorithm used
to resolve LSAP problems, which is used in each iteration of the IPFP
algorithm.


\section{Conclusion}

We have in this paper characterized the solutions of some LSAP and QAP
as edit paths belonging to a certain family.  This
characterization has allowed us to solve the graph edit distance
problem through a QAP with a clear definition of the family of edit
paths implied in the minimization process.

The proposed algorithm used to solve the QAP usually provides a
permutation matrix, and hence a value which can be directly
interpreted as an approximation of the GED.  If the final matrix is
only  stochastic, we simply project it on the set of
permutation matrices in order to obtain the desired value.

Our experiments show that the proposed algorithm usually find values
quite close from the optimal solutions within computational times
which allow to apply it on graphs of non trivial sizes. These
experiments also show the importance of the initialization step since a
good initialization both improves the final result and reduces the
number of iterations.

Our further works should investigate different directions: the
initialization step, the removal within our algorithm of the different
instances of $\epsilon$ values and finally the definition of
alternative QAP solvers. Such improvements should allow in a near
future to compute efficiently close approximations of the graph edit
distance on graphs composed of a thousand of nodes.


\footnotesize
\bibliographystyle{plain}
\bibliography{technical_report_ged}

\end{document}